%% file: main-dt.tex
\documentclass[a4paper,UKenglish,cleveref, autoref, thm-restate]{lipics-v2021}
\nolinenumbers
\usepackage{latexsym}
\input{Preamble}

\title{A Finite Algorithm for the Realizabilty of a Delaunay Triangulation}%{An extended abstract of this paper is to appear in the 17th International Symposium on Parameterized and Exact Computation (IPEC), 2022.} %TODO Please add

\author{Akanksha Agrawal}{Indian Institute of Technology Madras, India}{akanksha@cse.iitm.ac.in}{https://orcid.org/0000-0002-0656-7572}{Supported by New Faculty Initiation Grant no. NFIG008972.}%TODO mandatory, please use full name; only 1 author per \author macro; first two parameters are mandatory, other parameters can be empty. Please provide at least the name of the affiliation and the country. The full address is optional. Use additional curly braces to indicate the correct name splitting when the last name consists of multiple name parts.
\author{Saket Saurabh}{Institute of Mathematical Sciences, HBNI, India\\University of Bergen, Norway}{saket@imsc.res.in}{https://orcid.org/0000-0001-7847-6402}{Supported by European Research Council (ERC) under the European Union's Horizon 2020 research and innovation programme (no. 819416), and Swarnajayanti Fellowship (no. DST/SJF/MSA01/2017-18).\begin{minipage}{0.1\textwidth}
    \begin{center}
        \includegraphics[scale=0.5]{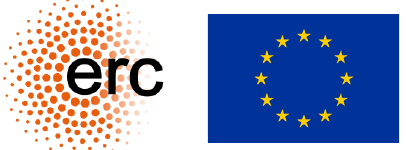}
    \end{center}
\end{minipage}}

\author{Meirav Zehavi}{Ben-Gurion University of the Negev, Israel}{meiravze@bgu.ac.il}{https://orcid.org/0000-0002-3636-5322}{Supported by Israel Science Foundation grant no. 1176/18, and United States – Israel Binational Science Foundation grant no. 2018302.\\
*{\em An extended abstract of this paper is to appear in the 17th International Symposium on Parameterized and Exact Computation (IPEC), 2022.}
}

\authorrunning{A. Agrawal and S. Saurabh and M. Zehavi} %TODO mandatory. First: Use abbreviated first/middle names. Second (only in severe cases): Use first author plus 'et al.'
\Copyright{Akanksha Agrawal and Saket Saurabh and Meirav Zehavi} %TODO mandatory, please use full first names. LIPIcs license is "CC-BY";  http://creativecommons.org/licenses/by/3.0/

 %TODO mandatory: Please choose ACM 2012 classifications from https://dl.acm.org/ccs/ccs_flat.cfm 

 %TODO mandatory; please add comma-separated list of keywords

\relatedversion{} %optional, e.g. full version hosted on arXiv, HAL, or other respository/website
\begin{document}
	\maketitle
%--------------Abstract---------------------------------------------
\begin{abstract}
\input{abstractDT}

	\end{abstract}

\input{intro.tex}

\input{prelims.tex}

\input{sub-sec-1-dr.tex}

\input{rational-coordinates-dt.tex}

\input{integer-coordinates-dt.tex}

\section{Conclusion}
In this paper, we gave an $n^{\OO(n)}$-time algorithm for the \delaunaytr{} problem. We have thus obtained the first exact exponential-time algorithm for this problem. Still, the existence of a practical (faster) exact algorithm for \delaunaytr{} is left for further research. In this context, it is not even clear whether a significantly faster algorithm, say a polynomial-time algorithm, exists. Perhaps one of the first questions to ask in this regard is whether there exist instances of graphs that are realizable but for which the integers in any integral solution need to be exponential in the input size? If yes, does even the representation of these integers need to be exponential in the input size?

\bibliography{references}

\end{document}

%% file: Preamble.tex
\usepackage{todonotes}
\newcommand{\delaunaytr}{\textsc{Delaunay Realization}}
\newcommand{\delaunaytrIntFull}{\textsc{Integral Delaunay Realization}}
\newcommand{\delaunaytrInt}{\textsc{Int-DR}}
\newcommand{\delaunaytrRatFull}{\textsc{Rational Delaunay Realization}}
\newcommand{\delaunaytrRat}{\textsc{Rational-DR}}

\newcommand{\connCDTfull}{\textsc{Restricted-Delaunay Realization}}
\newcommand{\connCDT}{\textsc{Res-DR}}
\newcommand{\triangulation}{triangulation}

\newcommand{\DT}{Delaunay triangulation}
\newcommand{\DTs}{Delaunay triangulations}
\newcommand{\DG}{Delaunay graph}
\newcommand{\DR}{Delaunay realizable}

\newcommand{\triangulatedgraph}{triangulated graph}
\newcommand{\fontDT}{\mathscr}

\newcommand{\OO}{\mathcal{O}}

\newcommand{\yes}{\textsc{YES}}
\newcommand{\no}{\textsc{NO}}

\newcommand{\defproblemout}[3]{
  \vspace{1mm}
\noindent\fbox{
  \begin{minipage}{0.96\textwidth}
  \begin{tabular*}{\textwidth}{@{\extracolsep{\fill}}lr} #1 \\ \end{tabular*}
  {\bf{Input:}} #2  \\
  {\bf{Output:}} #3
  \end{minipage}
  }
  \vspace*{2mm}
}

%% file: abstractDT.tex
The \emph{Delaunay graph} of a point set $P \subseteq \mathbb{R}^2$ is the plane graph with the vertex-set $P$ and the edge-set that contains $\{p,p'\}$ if there exists a disc whose intersection with $P$ is exactly $\{p,p'\}$. Accordingly, a triangulated graph $G$ is \emph{Delaunay realizable} if there exists a triangulation of the Delaunay graph of some $P \subseteq \mathbb{R}^2$, called a \emph{Delaunay triangulation} of $P$, that is isomorphic to $G$. The objective of \textsc{Delaunay Realization} is to compute a point set $P \subseteq \mathbb{R}^2$ that realizes a given graph $G$ (if such a $P$ exists). Known algorithms do not solve \textsc{Delaunay Realization} as they are non-constructive. Obtaining a constructive algorithm for \textsc{Delaunay Realization} was mentioned as an open problem by Hiroshima et al.~\cite{hiroshima2000}. We design an $n^{\mathcal{O}(n)}$-time constructive algorithm for \textsc{Delaunay Realization}. In fact, our algorithm outputs sets of points with {\em integer} coordinates. 
%The design of our algorithm involves a careful manipulation of a (hypothetical) point set in $\mathbb{R}^2$, which allows to argue that it is ``safe'' to add new polynomials to a preliminary set of polynomials devised to describe the problem. Having these new polynomials, we are able to ensure that certain approximate solutions, which we can find in finite time, are actually exact solutions.
%We believe that our contribution is a valuable step forward in the study of algorithms for geometric problems where one is interested in finding a solution rather than only determining whether one exists.%, which has so far been the focus in the study of many well-known geometric problems, such as the \textsc{Delaunay Realization} problem considered here.

%% file: intro.tex
\section{Introduction}\label{sec:intro}

We study Delaunay graphs---through the lens of the well-known \textsc{Delaunay Realization} problem---which are defined as follows.
Given a point set $P \subseteq \mathbb{R}^2$, the \emph{Delaunay graph, $\mathscr{DG}(P)$, of $P$} is the graph with vertex-set $P$ and edge-set that consists of every pair $(p,p')$ of points in $P$ that satisfies the following condition: there exists a disc whose boundary intersects $P$ only at $p$ and $p'$, and whose interior does not contain any point in $P$. The point set $P \subseteq \mathbb{R}^2$ is in general position if it contains no four points from $P$ on the boundary of a disc. If $P$ is in general position, $\mathscr{DG}(P)$ is a triangulation, called a \emph{Delaunay triangulation}, denoted by $\mathscr{DT}(P)$.\footnote{We assume that $|P| \geq 4$, as otherwise, the problem that we consider, is solvable in polynomial time.} Otherwise, Delaunay triangulation and the notation $\mathscr{DT}(P)$, may refer to any triangulation obtained by adding edges to $\mathscr{DG}(P)$. Thus, Delaunay triangulation of a point set $P$ is unique if and only if $\mathscr{DG}(P)$ is a triangulation. An alternate characterization of \DT s is that in such a triangulation, for any three points of a triangle of an interior face, the unique disc whose boundary contains these three points does not contain any other point in $P$.

The \DG\ of a point set is a planar graph~\cite{comp-geom-berg}, and triangulations of such graphs form an important subclass of the class of triangulations of a point set, also known as the class of maximal planar sub-divisions of the plane. Accordingly,  efficient algorithms for computing a \DT\ for a given point set have been developed (see~\cite{comp-geom-berg,Clarkson:1989,Guibas1992}). One of the main reasons underlying the interest in \DTs\ is that any angle-optimal triangulation of a point set is actually a \DT\ of the point set. Here, optimality refers to the maximization of the smallest angle~\cite{comp-geom-berg,DBLP:journals/siamdm/BattistaV96}.
This property is particularly useful when it is desirable to avoid ``slim'' triangles---this is the case, for example, when approximating a geographic terrain. Another main reason underlying the interest in \DTs\ is that these triangulations are the duals of ``Voronoi diagrams'' (see~\cite{Okabe-comp-geom-book}). 

We are interested in a well-known problem which, in a sense, is the ``opposite'' of computing a \DT\ for a given point set. Here, rather than a point set, we are given a triangulated graph $G$. The graph $G$ is \emph{Delaunay realizable} if there exists $P \subseteq \mathbb{R}^2$ such that $\mathscr{DT}(P)$ is isomorphic to $G$. Specifically, a point set $P \subseteq \mathbb{R}^2$ is said to {\em realize} $G$ (as a \DT) if $\mathscr{DT}(P)$ is isomorphic to $G$.\footnote{As $G$ is triangulation, if $\mathscr{DT}(P)$ is isomorphic to $G$, then $\mathscr{DT}(P)$ is unique.} The problem of finding a point set that realizes $G$ is called \textsc{Delaunay Realization}. This problem is important not only theoretically, but also practically (see, e.g.,~\cite{OISHI,sugihara1992const,sugihara1994robust}). Formally, it is defined as follows.

\defproblemout{\delaunaytr}{A \triangulation\ $G$ on $n$ vertices.}{If $G$ is realizable as a \DT, then output $P \subseteq \mathbb{R}^2$ that realizes $G$ (as a \DT). Otherwise, output \no.} 

Dillencourt~\cite{Dillencourt:1987:TDT} established necessary conditions for a triangulation to be realizable as a \DT. On the other hand, Dillencourt and Smith~\cite{Dillencourt-DT} established sufficient conditions for a triangulation to be realizable as a \DT. Dillencourt~\cite{Dillencourt:1990} gave a constructive proof showing that any triangulation where all vertices lie on the outer face is realizable as a \DT. Their approach, which results in an algorithm that runs in time $\OO(n^2)$, uses a criterion concerning angles of triangles in a hypothetical \DT. In 1994, Sugihara~\cite{Sugihara:1994} gave a simpler proof that all outerplanar triangulations are realizable as \DTs. Later, in 1997, Lambert~\cite{lambert1997} gave a linear-time algorithm for realizing an outerplanar triangulation as a \DT. More recently, Alam et al.~\cite{AlamRS11} gave yet another constructive proof for outerplanar triangulations. 

Hodgson et al.~\cite{hodgson1992char} gave a polynomial-time algorithm for checking if a graph is realizable as a convex polyhedron with all vertices on a common sphere. Using this, Rivin~\cite{rivin1994euclidean} designed a polynomial-time algorithm for testing if a graph is realizable as a \DT. Independently, Hiroshima et al.~\cite{hiroshima2000} found a simpler polynomial-time algorithm, which relies on the proof of a combinatorial characterization of Delaunay realizable graphs. Both these results are non-constructive, i.e., they cannot output a point set $P$ that realizes the input as a Delaunay triangulation, but only answer \yes{} or \no. It is a long standing open problem to design a finite time algorithm for \textsc{Delaunay Realization}. 
 
Obtaining a constructive algorithm for \textsc{Delaunay Realization} was mentioned as an open problem by Hiroshima et al.~\cite{hiroshima2000}. We give the {\em first} exponential-time algorithm for the \delaunaytr{} problem. Our algorithm is based on the computation of two sets of polynomial constraints, defined by the input graph $G$. In both sets of constraints, the degrees of the polynomials are bounded by $2$ and the coefficients are integers. The first set of constraints forces the points on the outer face to form a convex hull,\footnote{The convex hull of a point set realizing $G$ forms the outer face of its \DT.} and the second set of constraints ensures that for each edge in $G$, there is a disc containing only the endpoints of the edge. Roughly speaking, we prove that a triangulation is realizable as a \DT\ if and only if a point set realizing it as a \DT\ satisfies every constraint in our two sets of constraints. We proceed by proving that if a triangulation is realizable as a \DT, then there is $P \subseteq \mathbb{Z}^2$ such that $\mathscr{DT}(P)$ is isomorphic to $G$. This result is crucial to the design of our algorithm, not only for the sake of obtaining an integer solution, but for the sake of obtaining any solution. In particular, it involves a careful manipulation of a (hypothetical) point set in $\mathbb{R}^2$, which allows to argue that it is ``safe'' to add new polynomials to our two sets of polynomials. Having these new polynomials, we are able to ensure that certain approximate solutions, which we can find in finite time, are actually exact solutions. We show that the special approximate solutions can be computed in polynomial time, and hence we actually solve the problem precisely. To find a solution satisfying our sets of polynomial constraints, our algorithm runs in time $n^{\OO(n)}$. All other steps of the algorithm can be executed in polynomial~time.% The design of a practical (faster) exact exponential-time algorithm for \delaunaytr{} is an interesting direction for further research. 

%It is worth noting that not all geometric realization problems admit solutions with points of rational (or integer) coordinates. For instance, {\sc Line-Segment Graph Realization} is known to be $\exists \mathbb{R}$-complete~\cite{DBLP:journals/jct/KratochvilM94}. A more contrasting example to the fact that \delaunaytr{} admit solutions with rational (and integer) coordinates, is the result for its higher dimension analogue--- {\sc Delaunay subdivisions Realization} is $\exists \mathbb{R}$-complete~\cite{DBLP:journals/dcg/AdiprasitoPT15}. 

We believe that our contribution is a valuable step forward in the study of algorithms for geometric problems where one is interested in finding a solution rather than only determining whether one exists. Such studies have been carried out for various geometric problems (or their restricted versions) like {\sc Unit-Disc Graph Realization}~\cite{DBLP:journals/jct/McDiarmidM13}, {\sc Line-Segment Graph Realization}~\cite{DBLP:journals/jct/KratochvilM94}, {\sc Planar Graph Realization} (which is the same as {\sc Coin Graph Realization})~\cite{DBLP:conf/stoc/FraysseixPP88}, {\sc Convex Polygon Intersection Graph Realization}~\cite{DBLP:journals/siamdm/MullerLL13}, and \delaunaytr. (The above list is not comprehensive; for more details we refer the readers to given citations and references therein.) We note that the higher dimension analogue of \delaunaytr, called {\sc Delaunay Subdivisions Realization}, is $\exists \mathbb{R}$-complete; for details on this generalization, see~\cite{DBLP:journals/dcg/AdiprasitoPT15}. %Some details (and results marked with $\spadesuit$) are omitted due to space constraints (see full version for complete details).

%Due to limited space, proofs of (marked with $\spadesuit$) can be found in the appendix.

%which can be found in the attached full version.% many details have been relegated to the appendix.

 %A contrasting example to the fact that \delaunaytr{} admit solutions with rational (and integer) coordinates, is the result for its higher dimension analogue--- {\sc Delaunay subdivisions Realization} is $\exists \mathbb{R}$-complete~\cite{DBLP:journals/dcg/AdiprasitoPT15}. 

%\todo{newly added, do we need citations?}%We remark that proofs of results marked with $\spadesuit$ were relegated to the~appendix.\todo{remove.}

%% file: prelims.tex
\section{Preliminaries}
In this section, we present basic concepts related to Geometry, Graph Theory and Algorithm Design, and establish some of the notation used throughout. %We refer the reader to the books~\cite{comp-geom-berg,opac-b1093656} for geometry-related terms that are not explicitly defined here.

%Definitions related to Graph Theory, along with other basic notions, are given in Appendix~\ref{app:standard}. 
%\subsection{Standard Definitions and Notation}\label{app:standard}
%
We refer the reader to the books~\cite{comp-geom-berg,opac-b1093656} for geometry-related terms that are not explicitly defined here. 
We denote the set of natural numbers by $\mathbb{N}$, the set of rational numbers by $\mathbb{Q}$ and the set of real numbers by $\mathbb{R}$. By $\mathbb{R}^+$ we denote the set $\{x \in \mathbb{R} \mid x >0\}$. For $n \in \mathbb{N}$, we use $[n]$ as a shorthand for $\{1,2,\cdots, n\}$. A point is an element in $\mathbb{R}^2$. We work on Euclidean plane and the Cartesian coordinate system with the underlying bijective mapping of points in the Euclidean plane to vectors in the Cartesian coordinate system. For $p, q \in \mathbb{R}^2$, by $\textsf{dist}(p,q)$ we denote the distance between $p$ and $q$ in $\mathbb{R}^2$.

\subparagraph{Graphs.} We use standard terminology from the book of Diestel~\cite{diestel-book} for graph-related terms not explicitly defined here. For a graph $G$, $V(G)$ and $E(G)$ denote the vertex and edge sets of $G$, respectively. For a vertex $v \in V(G)$, $d_G(v)$ denotes the {\em degree} of $v$, i.e the number of edges incident on $v$, in the graph $G$. For an edge $(u,v) \in E(G)$, $u$ and $v$ are called the {\em endpoints} of the edge $(u,v)$. For $S \subseteq V(G)$, $G[S]$, and $G - S$ are the subgraphs of $G$ induced on $S$ and $V(G) \setminus S$, respectively. For $S \subseteq V(G)$, we let $N_G(S)$ and $N_G[S]$ denote the open and closed neighbourhoods of $S$ in $G$, respectively. That is, $N_G(S) = \{v \mid (u,v) \in E(G), u \in S\} \setminus S$ and $N_G[S] = N_G(S) \cup S$. We drop the sub-script $G$ from $d_G(v),N_{G}(S)$, and $N_{G}[S]$ whenever the context is clear. A {\em path} in a graph is a sequence of distinct vertices $v_0, v_1, \ldots, v_\ell$ such that $(v_i,v_{i+1})$ is an edge for all $0 \leq i < \ell$. Furthermore, such a path is called a $v_0$ to $v_\ell$ path. A graph is {\em connected} if for all distinct $u,v \in V(G)$, there is a $u$ to $v$ path in $G$. A graph which is not connected is said to be {\em disconnected}. A graph $G$ is called {\em $k$-connected} if for all $X \subseteq V(G)$ such that $|X| < k$, $G -X$ is connected. A {\em cycle} in a graph is a sequence of distinct vertices $v_0, v_1, \ldots, v_\ell$ such that $(v_i,v_{ (i+1) \mod (\ell+1)})$ is an edge for all $0 \leq i \leq \ell$. A cycle $C$ in $G$ is said to be a {\em non-separating cycle} in $G$ if $G-V(C)$ is connected.

\subparagraph{Planar Graphs and Plane Graphs.} A graph $G$ is called \emph{planar} if it can be drawn on the plane such that no two edges cross each other except possibly at their endpoints. Formally, an embedding of a graph $G$ is an injective function $\varphi: V(G) \rightarrow \mathbb{R}^2$ together with a set ${\cal C}$ containing a continuous curve $C_{(u,v)}$ in the plane corresponding to each $(u,v) \in E(G)$ such that $\varphi(u)$ and $\varphi(v)$ are the endpoints of $C_{(u,v)}$. An embedding of a graph $G$ is {\em planar } if distinct $C,C' \in {\cal C}$ intersect only at the endpoints---that is, any point in the intersection of $C,C'$ is an endpoint of both $C,C'$. A graph that admits a planar embedding is a {\em planar graph}. Hereafter, whenever we say an embedding of a graph, we mean a planar embedding of it, unless stated otherwise. We often refer to a graph with a fixed embedding on the plane as a {\em plane graph}. For a plane graph $G$, the regions in $\mathbb{R}^2 \setminus G$ are called the \emph{faces} of $G$. We denote the set of faces in $G$ by $F(G)$. Note that since $G$ is bounded and can be assumed to be drawn inside a sufficiently large disc, there is exactly one face in $F(G)$ that is unbounded, which is called the {\em outer face} of $G$. A face of $G$ that is not the outer face is called an {\em inner face} of $G$. An embedding of a planar graph with the property that the boundary of every face (including the outer face) is a convex polygon is called a {\em convex drawing}.
Below we state propositions related to planar and plane graphs that will be useful later.

\begin{proposition}[Proposition 4.2.5~\cite{diestel-book}] \label{face-triangle-plane-graph}
For a $2$-connected plane graph $G$, every face of $G$ is bounded by a cycle.
\end{proposition}

For a graph $G$ and a face $f \in F(G)$, we let $V(f)$ denote the set of vertices in the cycle by which $f$ is bounded. We often refer to $V(f)$ as the \emph{face boundary} of $f$.

\begin{proposition}[Proposition 4.2.10~\cite{diestel-book}]\label{face-boundaries-conn-graph}
For a $3$-connected planar graph, its face boundaries are precisely its non-separating induced cycles.
\end{proposition}

Note that from Proposition~\ref{face-boundaries-conn-graph}, for a $3$-connected planar graph and its planar embeddings $G_{\cal P}$ and $G_{{\cal P}'}$, it follows that $F(G_{\cal P})= F(G_{{\cal P}'})$. (In the above we slightly abused the notation, and think of the sets $F(G_{\cal P})$ and $F(G_{{\cal P}'})$ in terms of their bounding cycles, rather than the regions of the plane.) Hence, it is valid to talk about $F(G)$ for a $3$-connected planar graph $G$, even without knowing its embedding on the plane.

\begin{proposition}[Tutte's Theorem~\cite{Tutte-draw-graphs}, also see~\cite{DBLP:journals/acta/ChibaON85,DBLP:journals/ieicet/NishizekiMR04}]\label{convex-drawing-plane-graph}
A $3$-connected planar graph admits a convex embedding on the plane with any face as the outer face. Moreover, such an embedding can be found in polynomial time.
\end{proposition}

For a plane graph $G$ and a face $f \in F(G)$, by \emph{stellating} $f$ we mean addition of a new vertex $v^*_f$ inside $f$ and making it adjacent to all $v \in V(f)$.  We note that stellating a face of a planar graph results in another planar graph~\cite{Dillencourt-DT}.

\subparagraph{Triangulations and Delaunay Triangulations.} A \emph{triangulation} is a plane graph where each inner face is bounded by a cycle on three vertices. A graph which is isomorphic to a triangulation is called a \emph{\triangulatedgraph}. We state the following simple but useful property of triangulations that will be exploited later.

\begin{proposition}\label{triangulation-properties}
Let $G$ be a triangulation with $f^*$ being the outer face. Then, all the degree-$2$ vertices in $G$ must belong to $V(f^*)$.
\end{proposition}

%Let $P \subseteq \mathbb{R}^2$ be a set of $n$ points. Then, the {\em Delaunay graph} $\fontDT{DG}(P)$ of $P$ is the graph with the vertex set  $P$ and the edge-set that consists of every pair $(p,p')$ of points in $P$ that satisfies the following condition: there exists a disc in $\mathbb{R}^2$ whose boundary intersects $P$ only at $p$ and $p'$, and whose interior does not contain any point in $P$~\cite{comp-geom-berg}.

\begin{proposition}[Theorem 9.6~\cite{comp-geom-berg}]\label{face-characterization-DT}
For a point set $P \subseteq \mathbb{R}^2$ on $n$ points, three points $p_1 , p_2 , p_3 \in P$ are vertices of the same face of the Delaunay graph of $P$ if and only if the circle through $p_1, p_2, p_3$ contains no point of $P$ in its interior.
\end{proposition}

A {\em \DT}\ is any triangulation that is obtained by adding edges to the Delaunay graph. A \DT\ of a point set $P$ is unique if and only if $\fontDT{DG}(P)$ is a triangulation, which is the case if $P$ is in general position~\cite{comp-geom-berg}. We refer to the \DT\ of a point set $P$ by $\fontDT{DT}(P)$ (assuming it is unique, which is the case in our paper). A \emph{triangulated graph} is \emph{\DR} if there exists a point set $P \subseteq \mathbb{R}^2$ such that $\fontDT{DT}(P)$ is isomorphic to $G$. If $G$ has at most three points, then testing if it is \DR\ is solvable in constant time. Also, we can compute an integer representation for it in constant time, if it exists. (Recall that while defining the general position assumption, we assumed that the point set has at least four points. This assumption does not cause any issues because we look for a realization of a graph which has at least four vertices.) 

\subparagraph{Polynomial Constraints.} Let us now give some definitions and notation related to polynomials and sets of polynomial constraints (equalities and inequalities). We refer the reader to the books~\cite{Artin:1998,book-real-algebraic-geometry-basu} for algebra-related terms that are not explicitly defined here. For $t,n \in \mathbb{N}$ and a set $C$, a polynomial $\mathscr{P}=\Sigma_{i \in [t]} a_i \cdot (\Pi_{j \in [n]} X_j^{d^i_j})$ on $n$ variables and $t$ terms is said to be a polynomial over $C$ if for all $i \in [t]$, $j \in [n]$ we have $a_i \in C$ and $d^i_j \in \mathbb{N}$. Furthermore, the {\em degree} of the polynomial $\mathscr{P}$ is defined to be $\max_{i \in [t]} (\Sigma_{j \in [n]} {d^i_j})$. We denote the set of polynomials on $n$ variables $X_1,X_2,\cdots X_n$ with coefficients in $C$ by $C[X_1,X_2,\cdots X_n]$.

A {\em polynomial constraint} $\mathscr{C}$ on $n$ variables with coefficients from $C\subseteq\mathbb{R}$ is a sequence $\mathscr{P} \Delta 0$, where $\mathscr{P} \in C[X_1,X_2,\cdots, X_n]$  and $\Delta \in \{=, \geq, >, \leq , <\}$. The {\em degree} of such a constraint is the degree of $\mathscr{P}$, and it is said to be an equality constraint if $\Delta$ is `='. We say that the constraint is {\em satisfied} by an element $(\bar x_1, \bar x_2, \cdots \bar x_n)\in\mathbb{R}^n$ if  $\mathscr{P}(\bar x_1, \bar x_2, \cdots \bar x_n) \Delta 0$.\footnote{Here, $\mathscr{P}(\bar x_1, \bar x_2, \cdots \bar x_n)$ is the evaluation of $\mathscr{P}$, where the variable $X_i$ is assigned the value $\bar x_i$, for $i \in [n]$.}
Given a set $\mathfrak{C}$ of polynomial constraints on $n$ variables, $X_1,X_2,\cdots, X_n$, and with coefficients from $C\subseteq\mathbb{R}$, we say that an element $(\bar x_1, \bar x_2, \cdots \bar x_n) \in \mathbb{R}^n$ {\em satisfies} $\mathfrak{C}$ if for all $\mathscr{C} \in \mathfrak{C}$, we have that $(\bar x_1, \bar x_2, \cdots \bar x_n)$ satisfies $\mathscr{C}$. In this case, $(\bar x_1, \bar x_2, \cdots \bar x_n)$ is also called a {\em solution} of  $\mathfrak{C}$.
Furthermore, $\mathfrak{C}$ is said to be {\em satisfiable} (in $\mathbb{R}$) if there exists $(\bar x_1, \bar x_2, \cdots \bar x_n) \in \mathbb{R}^n$ satisfying~$\mathfrak{C}$. 

%Given a set $\mathfrak{P}$ of $m$ polynomials over $\mathbb{R}$ on $n$ variables with coefficients in $\mathbb{Z}$, a solution (or a zero) to $\mathfrak{P}$ refers to $(\bar x_1, \bar x_2, \cdots \bar x_n) \in \mathbb{R}^n$ such that for each $\mathscr{P} \in \mathfrak{P}$, $\mathscr{P}(\bar x_1, \bar x_2, \cdots \bar x_n) = 0$. Observe that for a solution $(\bar x_1, \bar x_2, \cdots \bar x_n) \in \mathbb{R}^n$ to $\mathfrak{P}$, it holds that $\bar x_i \in \mathbb{R}$ for $i \in [n]$. %Since we cannot represent real numbers in a computer, we will be using a uni-variate polynomial representation of each $\bar x_i$, $i \in [n]$, whose existence is guaranteed by Theorem 13.17 in the book by Basu et al.~\cite{book-real-algebraic-geometry-basu}.
Below we state a result regarding a method for solving a finite set of polynomial constraints, which will be used by our algorithm. This result is a direct implication of Propositions 3.8.1 and 4.1 in \cite{Renegar:1992} (see also \cite{book-real-algebraic-geometry-basu}).

\begin{proposition}[Propositions 3.8.1 and 4.1 in \cite{Renegar:1992}]\label{prop:solving-equations}
Let $\mathfrak{C}$ be a set of $m$ polynomial constraints of degree $2$ on $n$ variables with coefficients in $\mathbb{Z}$ whose bitsizes are bounded by $\OO(1)$. Then, in time $m^{\OO(n)}$ we can decide if $\mathfrak{C}$ is satisfiable in $\mathbb{R}$. Moreover, if $\mathfrak{C}$ is satisfiable in $\mathbb{R}$, then in time $m^{\OO(n)}$ we can also compute a (satisfiable) set $\widehat{\mathfrak{C}}$ of $n$ polynomial constraints, $\mathscr{C}_1,\mathscr{C}_2,\ldots,\mathscr{C}_n$, with coefficients in $\mathbb{Z}$, where for all $i\in[n]$, we have that $\mathscr{C}_i$ is an equality constraint on $X_i$ (only), and a solution of $\widehat{\mathfrak{C}}$ is also a solution of $\mathfrak{C}$.
\end{proposition}

%% file: sub-sec-1-dr.tex
\section{\connCDTfull: Generating Polynomials}\label{sec:3-conn-DR}
In this section, we generate a set of polynomials that encodes the realizability of a triangulation as a \DT\ in the case where the outer face of the \DT\ is known. More precisely, we suppose that the outer faces of $G$ and the \DT\ are the same. For the general case where we might not know a priori which is the face in $G$ that is supposed to be the outer face of the \DT\ (this is the case when $G$ is a maximal planar graph), we will ``guess'' the outer face and then use our restricted version to solve the problem. Formally, we solve the following problem. 

\defproblemout{\connCDTfull\ (\connCDT)}{A \triangulation\ $G$ with outer face $f^*$.}{A set of polynomial constraints $\textsf{Const}(G)$ such that $\textsf{Const}(G)$ is satisfiable if and only if $G$ is realizable as a \DT\ with $f^*$ as the outer face.}

Let $(G,f^*)$ be an instance of \connCDT, and let $n$ denote $|V(G)|$. We denote $V(G)$ by the set $\{v_1,v_2, \cdots v_n\}$. Note that except possibly $f^*$, each of the faces of $G$ is bounded by a cycle on three vertices. With each $v_i \in V(G)$ we associate two variables, $X_i$ and $Y_i$, which correspond to the values of the $x$ and $y$ coordinates of $v_i$ in the plane. Furthermore, we let $P_i$ denote the vector $(X_i,Y_i)$.
We let $\bar X$ denote the value that some solution of $\textsf{Const}(G)$ assigns to the variable $X$. Accordingly, we denote $\bar{P}_i=(\bar{X}_i,\bar{Y}_i)$. For the sake of clarity, we sometimes abuse the notation $\bar{P}_i$ by letting it denote both $\bar{P}_i$ and $P_i$  (this is done in situations where both interpretations are valid).

Our algorithm is based on the computation of two sets of polynomial constraints of bounded degree and integer coefficients. Informally, we have one set of inequalities which ensures that the points to which vertices of $f^*$ are mapped are in convex position, and another set of inequalities which ensures that for each $(v_i,v_j) \in E(G)$, there exists a disc containing $(\bar X_i,\bar Y_i)$ and $(\bar X_j,\bar Y_j)$ on its boundary and excluding all other points $(\bar X_{k}, \bar Y_k)$. (While other sets of inequalities may be devised to ensure these properties, we subjectively found the two sets presented here the easiest to employ.) 

\subsection{Inequalities Ensuring that the Outer Face Forms the Convex Hull}\label{sec:convexHull}

We first generate the set of polynomial constraints ensuring that the points associated with the vertices in $f^*$ form the convex hull of the output point set. Here, we also ensure that the vertices in $f^*$ have the same cyclic ordering (given by the cycle bounding $f^*$) as the points corresponding to them have in the convex hull. Note that the edges of the convex hull are present in any Delaunay triangulation~\cite{comp-geom-berg}. Moreover, the convex hull of a point set forms the outer face of its \DT. To formulate our equations, we rely on the notions of \emph{left and right turns}. Their definitions are the same as those in the book~\cite{book-algorithms-cormen}, which uses \emph{cross product} to determine whether a turn is a left turn or a right turn. For the sake of clarity, we also explain these notions below.

\subparagraph{Left and Right Turns.} Consider two vectors (or points) $\bar{P}_1$ and $\bar{P}_2$, denoting some $(x_1,y_1)$ and $(x_2,y_2)$, respectively. The cross product $\bar{P}_1 \times \bar{P}_2$ of $\bar{P}_1$ and $\bar{P}_2$ is defined as~follows. 

$ \bar{P}_1 \times \bar{P}_2 = \begin{vmatrix}
x_1& x_2\\
y_1& y_2
\end{vmatrix} = x_1y_2 - x_2y_1.$

If $\bar{P}_1 \times \bar{P}_2 > 0$, then $\bar{P}_1$ is said to be {\em clockwise} from $\bar{P}_2$ (with respect to the origin $(0,0)$). Else, if $\bar{P}_1 \times \bar{P}_2 < 0$, then $\bar{P}_1$ is said to be {\em counterclockwise} from $\bar{P}_2$. Otherwise (if $\bar{P}_1 \times \bar{P}_2 = 0$), $\bar{P}_1$ and $\bar{P}_2$ are said to be {\em collinear}. Given line segments $\overline{P_0P_1}$ and $\overline{P_1P_2}$, we would like to determine the type of turn taken by the angle $\angle P_0P_1P_2$. To this end, we check whether the directed segment $\overline{P_0P_2}$ is clockwise or counterclockwise from $\overline{P_0P_1}$. Towards this, we first compute the cross product $(\bar{P}_2- \bar{P}_0) \times (\bar{P}_1- \bar{P}_0)$. If $(\bar{P}_2- \bar{P}_0) \times (\bar{P}_1- \bar{P}_0)>0$, then $\overline{P_0P_2}$ is clockwise from $\overline{P_0P_1}$, and we say that we take a \emph{right turn} at $\bar{P}_1$. Else, if $(\bar{P}_2- \bar{P}_0) \times (\bar{P}_1- \bar{P}_0)<0$, then $\overline{P_0P_2}$ is counterclockwise from $\overline{P_0P_1}$, and we say that we take a \emph{left turn} at $\bar P_1$. Otherwise, we make {\em no turn} at $\bar P_1$. Note that the computation of $(\bar{P}_2- \bar{P}_0) \times (\bar{P}_1- \bar{P}_0)$ can be done as follows.

$(\bar{P}_2- \bar{P}_0) \times (\bar{P}_1- \bar{P}_0) = \begin{vmatrix}
x_2-x_0& x_1-x_0\\
y_2-y_0& y_1-y_0
\end{vmatrix} = x_2y_1 - x_2y_0 - x_0y_1-x_1y_2 + x_1y_0 + x_0y_2.$

\subparagraph{The Polynomials.}
For three vectors (or points) $\bar{P_0}=(x_0,y_0),\bar{P_1}=(x_1,y_1)$ and $\bar{P_2}=(x_2,y_2)$, by $\textsf{Con}(\bar{P_0},\bar{P_1},\bar{P_2})$ we denote the polynomial $x_2y_1 - x_2y_0 - x_0y_1-x_1y_2 + x_1y_0 + x_0y_2$. Note that $\textsf{Con}(\bar{P_0},\bar{P_1},\bar{P_2})$ determines whether we have a right, left or no turn at $\bar{P}_1$.

Before stating the constraints based on these polynomials, let us recall the well-known fact stating that a non-intersecting polygon is convex if and only if every interior angle of the polygon is less than $180^\circ$. While we ensure the non-intersecting constraint later, the characterization of each angle being less than $180^\circ$ is the same as taking a \emph{right} (or \emph{left}) turn at $P_j$ for every three consecutive points $P_i,P_j$ and $P_k$ of the polygon. We will use this characterization to enforce convexity on the points corresponding to the vertices in $V(f^*)$. Let us also recall that $f^*$ is a cycle $C^*$ in $G$. Next, whenever we talk about consecutive vertices in $C^*$, we always follow clockwise direction.

For every three consecutive vertices $v_i,v_j$ and $v_k$ in $C^*$, we add the following inequality:
$$\textsf{Con}(P_i,P_j,P_k)>0.$$
These inequalities ensure that in any output point set, the points corresponding to vertices in $V(f^*)$ are in convex position (together with the non-intersecting condition to be ensured later). 

Next, we further need to ensure that all the points which correspond to vertices in $V(G) \setminus V(f^*)$ belong to the interior of the convex hull formed by the points corresponding to vertices in $V(f^*)$ and the polygon formed by the points corresponding to $V(f^*)$ is non-self intersecting. For this purpose, we crucially rely on the following property of convex hulls (or convex polygons): For any edge of the convex hull, it holds that all the points, except for the endpoints of the edge, are located in one of the sides of the edge. Using this property, we know that for any two consecutive vertices $v_i$ and $v_j$ in $C^*$, all points are on one side of the line associated with $v_i$ and $v_j$. Since at each $v_i \in C^*$ we ensure that we turn right, we must have all the points located on the right of the line defined by the edge $(v_i,v_j)$. This, in turn, implies that for every pair of consecutive vertices $v_i$ and $v_j$ in $C^*$, for any vertex $v_k \in V(G) \setminus V(f^*)$, we must be turning left at $v_k$ (according to the ordered triplet $(v_i,v_k,v_j)$). Hence, we add the following inequalities:
$$\textsf{Con}(P_i,P_k,P_j)<0.$$
where $v_i$ and $v_j$ are consecutive vertices of $C^*$ and $v_k \in V(G) \setminus \{v_i,v_j\}$.

We denote the set of inequalities generated above by $\textsf{Con}(G)$.

\subsection{Inequalities Guaranteeing Existence of Edges}

For each edge $(v_i,v_j) \in E(G)$, we add two new variables, $X_{ij}$ and $Y_{ij}$, to indicate the coordinates of the centre of a disc that realizes the edge $(v_i,v_j)$. There might exist many discs that realize the edge $(v_i,v_j)$, but we are interested in only one such disc, say $C_{ij}$. Note that $C_{ij}$ should contain $(\bar X_i,\bar Y_i)$ and $(\bar X_j,\bar Y_j)$ on its boundary, and it should not contain any $(\bar X_k,\bar Y_k)$ such that $k \notin \{i,j\}$. Towards this, for each edge $(v_i,v_j) \in E(G)$, we add a set of inequalities that we denote by $\textsf{Dis}(v_i,v_j)$. Note that the radius $r_{ij}$ of $C_{ij}$ is given by $r^2_{ij} = (X_i-X_{ij})^2+(Y_i-Y_{ij})^2$ (if $(X_i,Y_i)$ lies on the boundary) and by $r^2_{ij} = (X_j-X_{ij})^2+(Y_j-Y_{ij})^2$ (if $(X_j,Y_j)$ lies on the boundary). Therefore, we want to ensure the following.

$$(X_j-X_{ij})^2+(Y_j-Y_{ij})^2=(X_i-X_{ij})^2+(Y_i-Y_{ij})^2$$
$$\Rightarrow X_i^2 - X_j^2 + Y_i^2 - Y_j^2 - 2X_{ij}X_i - 2Y_{ij}Y_i + 2X_{ij}X_j + 2Y_{ij}Y_j = 0.$$

Hence, we add the above constraint to $\textsf{Dis}(v_i,v_j)$. Further, we want to ensure that for each $k \in [n] \setminus \{i,j\}$, $(X_k,Y_k)$ does not belong to $C_{ij}$. Therefore, for each $k \in [n] \setminus \{i,j\}$, the following must hold.

$$(X_k-X_{ij})^2+(Y_k-Y_{ij})^2-(X_i-X_{ij})^2-(Y_i-Y_{ij})^2 >0$$
$$\Rightarrow X_k^2 - X_i^2 + Y_k^2 - Y_i^2 - 2X_{ij}X_k - 2Y_{ij}Y_k + 2X_{ij}X_i + 2Y_{ij}Y_i > 0$$

Hence, we also add the above constraint to $\textsf{Dis}(v_i,v_j)$ for $k \in [n] \setminus \{i,j\}$. Overall, we denote $\textsf{Dis}(G)=\displaystyle{\bigcup_{(v_i,v_j) \in E(G)}\textsf{Dis}(v_i,v_j)}$. This completes the description of all inequalities relevant to this section.

\subsection{Correctness}
Let us denote $\textsf{Const}(G)= \textsf{Con}(G) \cup \textsf{Dis}(G)$. We begin with the following observation. Here, to bound the number of variables, we rely on the fact that $G$ is a planar graph, its number of edges is upper bounded by $3n$, and hence in total we introduced less than $8n$ variables.

\begin{observation}\label{obs:ConstProperties}
The number of constraints in $\textsf{Const}(G)$ is bounded by $\OO(n^2)$ and the total number of variables is bounded by $\OO(n)$. Moreover, each constraint in $\textsf{Const}(G)$ is of degree 2, and its coefficients belong to $\{-2,-1,0,1,2\}$.
\end{observation}

Now, we state the central lemma establishing the correctness of our algorithm for~\connCDT. 

\begin{lemma}
\label{lem:correctness-dr-fwd-3-conn}
A triangulation $G$ with outer face $f^*$ is realizable as a \DT\ with $f^*$ as its outer face if and only if $\textsf{Const}(G)$ is satisfiable.
\end{lemma}
\begin{proof}
Let $G$ be a triangulation realizable as a \DT\ with $f^*$ as the outer face of the  \DT. Then, there exists $P \subseteq \mathbb{R}^2$ such that $\mathscr{DT}(P)$ is isomorphic to $G$ and $f^*$ is the outer face of $\mathscr{DT}(P)$. Furthermore, for each $(\bar P_i,\bar P_j) \in E(\mathscr{DT}(P))$, there exists a disc $C_{ij}$ which contains $\bar P_i$ and $\bar P_j$ on its boundary, and which contains no point $\bar P_k$, $k \in [n] \setminus \{i,j\}$, on neither its boundary nor its interior. We let $\bar P_{ij}$ denote the centre of $C_{ij}$. Let $\cal P$ be the vector assigning $\bar P_i$ to the vertex $v_i \in V(G)$ and $\bar P_{ij}$ to the centre of the disc $C_{ij}$. We note that the vertices of $f^*$ are in convex position in $\mathscr{DT}(P)$. Clearly, we then have that  ${\cal P}$ satisfies $\textsf{Const}(G)$. This concludes the proof of the forward direction.

In the reverse direction, consider some ${\cal P}$ that satisfies $\textsf{Const}(G)$. By our polynomial constraints, ${\cal P}$ assigns some $\bar P_i$ to each vertex $v_i \in V(G)$, such that for each edge $(v_i,v_j) \in E(G)$, it lets $\bar P_{ij}$ be the centre of a disc $C_{ij}$ containing $\bar P_i$ and $\bar P_j$ (on its boundary) and no point $\bar P_k$ where $k \notin \{i,j\}$. Further, we let $P= \{\bar P_i \mid i \in [n]\}$. By the construction of $\textsf{Const}(G)$, it follows that if $(v_i,v_j) \in E(G)$, then $(\bar P_i, \bar P_j) \in \mathscr{DT}(P)$ and the points in $P$ corresponding to vertices in $V(f^*)$ form the convex hull of $P$. This implies that the points corresponding to the vertices in $V(f^*)$ are on the outer face of $\mathscr{DT}(P)$. From Theorem 9.1 in~\cite{comp-geom-berg}, it follows that $|E(G)| \leq  |E(\mathscr{DT}(P))|$. Thus, $E(G) = E(\mathscr{DT}(P))$. This concludes the proof of the reverse direction.
  \end{proof}

%We thus derive the following result.
The next theorem follows from the construction of $\textsf{Const}(G)$, Observation \ref{obs:ConstProperties} and Lemma~\ref{lem:correctness-dr-fwd-3-conn}.

\begin{theorem}\label{thm:output-polynmials-dt}
Let $G$ be a triangulation on $n$ vertices with $f^*$ as the outer face. Then, in time $\OO(n^2)$, we can output a set of polynomial constraints $\textsf{Const}(G)$ such that $G$ is realizable as a \DT\ with $f^*$ as its outer face if and only if $\textsf{Const}(G)$ is satisfiable. Moreover, $\textsf{Const}(G)$ consists of $\OO(n^2)$ constraints and $\OO(n)$ variables, where each constraint is of degree 2 and with coefficients only from $\{-2,-1,0,1,2\}$.
\end{theorem}
%\begin{proof}
%The proof of the theorem follows directly from the construction of $\textsf{Const}(G)$, Observation \ref{obs:ConstProperties} and Lemma~\ref{lem:correctness-dr-fwd-3-conn}.
%  \end{proof}

%% file: rational-coordinates-dt.tex
\section{\connCDTfull: Replacing Points by Discs}\label{sec:exists-circle-dt}

Let $G$ be a \triangulation\ on $n$ vertices with $f^*$ as its outer face. Suppose that $G$ is realizable as a \DT\ where the points corresponding to vertices in $V(f^*)$ belong to the outer face. By Theorem~\ref{thm:output-polynmials-dt}, it follows that $\textsf{Const}(G)$ is satisfiable. Let $n^*$ denote the number of variables of $\textsf{Const}(G)$. % We note here that we use same notations as described in Section~\ref{sec:3-conn-DR}.
Since $\textsf{Const}(G)$ is satisfiable, there exists ${\cal Q}$ satisfying $\textsf{Const}(G)$. Let $\bar Q_i$ be the value assigned to the vertex $v_i \in V(G)$ for $i\in[n]$. Let $Q=\{\bar Q_i \mid v_i \in V(G)\}$. Recall that apart from assigning points in the plane to vertices in $V(G)$, ${\cal Q}$ assigns to each $(v_i,v_j) \in E(G)$, a point $\bar Q_{ij}$ corresponding to the centre of some disc, say $C'_{ij}$, containing $\bar Q_i,\bar Q_j$ on its boundary and excluding all other points in $Q$. 

In this section, we prove that for any given $\beta \in \mathbb{R}^+$, there exists a set of discs of radius $\beta$, one for each vertex in $V(G)$, with the following property. If for every $v_i\in V(G)$, we choose some point $\bar P_{C_i}$ inside or on the boundary of its disc $C_i$, we get that $\mathscr{DT}(Q)$ and the \DT\ of our set of chosen points are isomorphic.

We start with two simple observations, where the second directly follows from the definition of the constraints in $\textsf{Const}(G)$.

\begin{observation}\label{obs:scaling-distance}
Let $(a,b), (x,y) \in \mathbb{R}^2$ be two points and $\alpha \in \mathbb{R}^+$. Then, $\textsf{dist}((\alpha a,\alpha b),$ $(\alpha x,\alpha y)) = \alpha \cdot \textsf{dist}((a,b),(x,y))$.
\end{observation}

\begin{observation}\label{obs:multiply-by-const-soln}
Let $G$ be a triangulation on $n$ vertices with $f^*$ as its outer face. If ${\cal Q}$ is a solution of $\textsf{Const}(G)$, then for any $\alpha \in \mathbb{R}^+$, it holds that $\alpha {\cal Q}$ also satisfies $\textsf{Const}(G)$.
\end{observation}

%From observation~\ref{obs:multiply-by-const-soln}, it implies that $3\beta P \in \textsf{Zer}(\textsf{ConstEq}(G,f^*),\mathbb{R}^{n^*})$, where $3\beta P = \{(3\beta \bar X, 3\beta \bar Y) \mid (\bar X, \bar Y) \in P\}$. The constant $3\beta$ will be justified in the later part of the section.
In what follows, we create a point set $P$ such that $\mathscr{DT}(P)$ is isomorphic to $G$, where the points corresponding to vertices in $V(f^*)$ form the outer face of $\mathscr{DT}(P)$. We then show that this point set defines a set of discs with the desired property---for each  $v_i \in V(G)$, it defines one disc $C_i$ with $\bar P_i$ as centre and with radius $r^* \geq \beta > 0$ (to be determined), such that, roughly speaking, each point of $C_i$ is a valid choice for $v_i$.
For this purpose, we first define the real numbers, $d_N$, $d_C$, and $d_A$, which are necessary to determine $r^*$ and $P$. Informally, $d_N$ ensures that the discs we create around  vertices do not intersect, $d_C$ will be used to ensure existence of specific edges, $d_A$ will be used to ensure that ``convex hull property'' is satisfied. These (positive) real numbers are defined as follows.%We move to the description of $d_N$, $d_C$, $d_L$ and $d_A$.
\begin{itemize}
\item Let $\displaystyle{d_N=\min_{i,j \in [n], i \neq j} \{ \textsf{dist}(\bar Q_i, \bar Q_j) \}}$, i.e., $d_N$ is the minimum distance between any pair of distinct points in $Q$. 

\item Let $\displaystyle{d_C = \min_{i,j,k \in [n], i\neq j, i\neq k, j\neq k}\{\textsf{dist}(C'_{ij}, \bar Q_k)\}}$, i.e., $d_C$ denotes the minimum distance between a point corresponding to a vertex in $V(G)$ and a disc realizing an edge non-incident to it. (Recall that $C'_{ij}$ is defined at the beginning of this section.) Note that $d_C>0$ because in the above definition of  $d_C$, we have only considered those disc and point pairs where the point lies outside the disc.

%\item Towards the definition of $d_L$, consider some vertex $v_j \in V(f^*)$ with $v_i,v_k$ being the neighbors of $v_j$ on the cycle corresponding to the outer face $f^*$, i.e., $v_i,v_j$ and $v_k$ are consecutive vertices of $f^*$. Let $L_{ik}$ be the line containing the points $\bar Q_i$ and $\bar Q_k$, and let $d_j$ denote the distance between $\bar Q_j$ and $L_{ik}$. In this context, recall that $\mathcal{Q}$ satisfies $\textsf{Con}(G)$. In particular, it holds that $\bar Q_j$ does not lie on the line $L_{ik}$.  Let $\displaystyle{d_L= \min_{v_j \in V(f^*)}\{d_j\}}$.

\item For each edge $(v_i,v_j)$ of the cycle corresponding to the outer face $f^*$, let $L^s_{ij}$ be the line containing $\bar Q_i$ and $\bar Q_j$. Moreover, let $s_{ij} = \displaystyle{\min_{k \in [n] \setminus \{i,j\}} \{\textsf{dist}(L^s_{ij},\bar Q_k)\}}$, i.e., the minimum distance between a line of the convex hull and another point. Finally, $\displaystyle{d_A= \min_{(v_i,v_j) \in E(f^*)}\{s_{ij}\}}$. We note that $\displaystyle{d_A}>0$. This follows from the definition of $\textsf{Con}(G)$ in Section~\ref{sec:convexHull}. 
\end{itemize}

Define $r= \frac{1}{3}\min\{d_N,d_C,d_A\}$. Notice that $r, \beta > 0$. Now, we compute $r^*$ and $P$ according to three cases:

\begin{enumerate}
\item If $r \geq \beta$, then  $r^*=r$ and ${\cal P}= {\cal{Q}}$ (thus, $P=Q$).

\item Else if $1\leq r < \beta$, then ${\cal P}= \beta {\cal Q}$, where $\beta {\cal Q} = \{(\beta \bar X, \beta \bar Y) \mid (\bar X, \bar Y) \in {\cal Q}\}$ and $r^*=\beta r$. %Note that Observation~\ref{obs:multiply-by-const-soln} implies that $\beta {\cal Q} \in \textsf{Zer}(\textsf{ConstEq}(G),\mathbb{R}^{n^*})$.

\item Otherwise ($r < 1$ and $r < \beta$), ${\cal P} = \frac{\beta}{r} {\cal Q}$ and $r^*=\frac{\beta}{r}r= \beta$.

\end{enumerate}

By Observation~\ref{obs:multiply-by-const-soln}, in each of the cases described above, we have that ${\cal P}$ satisfies $\textsf{Const}(G)$. Hereafter, we will be working only with ${\cal P}$ and $r^*$ as defined above. We let $\bar P_i$ be the point assigned to the vertex $v_i$, and $P=\{\bar P_i \mid i \in [n]\}$. Moreover, we let $\bar P_{ij}$ be the centre of the disc $C_{ij}$ for the edge $(v_i,v_j) \in E(G)$ that is assigned by ${\cal P}$.

Next, we define $d^*_N,d^*_C,$ and $d^*_A$ in a manner similar to the one used to define $d_N,d_C$ and $d_A$. Let $d^*_N = \displaystyle{\min_{i,j \in [n], i\neq j}\{\textsf{dist}(\bar P_i, \bar P_j)\} \geq 3r^*}$, and $d^*_C = \displaystyle{\min_{i,j,k \in [n], i\neq j, i\neq k, j\neq k}}\{\textsf{dist}(C_{ij}, \bar P_k)\}$ $\geq 3r^*$. %To define $d^*_L$, consider a vertex $v_j \in V(f^*)$ with $v_i$ and $v_k$ being the neighbors of $v_j$ on the cycle corresponding to the outer face $f^*$. Let $L_{ik}$ be the line containing the points $\bar P_i$ and $\bar P_k$, and let $d^*_j$ be the distance between $\bar P_j$ and the line $L_{ik}$. Then, let $d^*_L= \displaystyle{\min_{v_j \in V(f^*)}\{d^*_j\}}$. 
 For each edge $(v_i,v_j)$ of the cycle corresponding to the outer face $f^*$, let $L^s_{ij}$ be the line containing $\bar P_i$ and $\bar P_j$. Further, let $s_{ij} = \displaystyle{\min_{k \in [n] \setminus \{i,j\}} \{\textsf{dist}(L^s_{ij},\bar P_k)\}}$. Finally, let $d^*_A= \displaystyle{\min_{(v_i,v_j) \in E(F)}\{s_{ij}\}}$. Note that by Observation~\ref{obs:scaling-distance}, we have that $d^*_N \geq 3r^*$, $d^*_C \geq 3r^*$, and $d^*_A \geq 3r^*$.

For each $v_i \in V(G)$, let $C_i$ be the disc of radius $r^*$ and centre $\bar P_i$. We now prove that if for each vertex $v_i \in V(G)$, we choose a point $\bar P'_i$ inside or on the boundary of $C_i$, then we obtain a point set $P'$ such that $\mathscr{DT}(P)$ and $\mathscr{DT}(P')$ are isomorphic. Furthermore, the points on the outer face of $\mathscr{DT}(P)$, and also $\mathscr{DT}(P')$, correspond to the vertices in~$V(f^*)$.

\begin{figure}[t]
\centering
\includegraphics[scale=0.21]{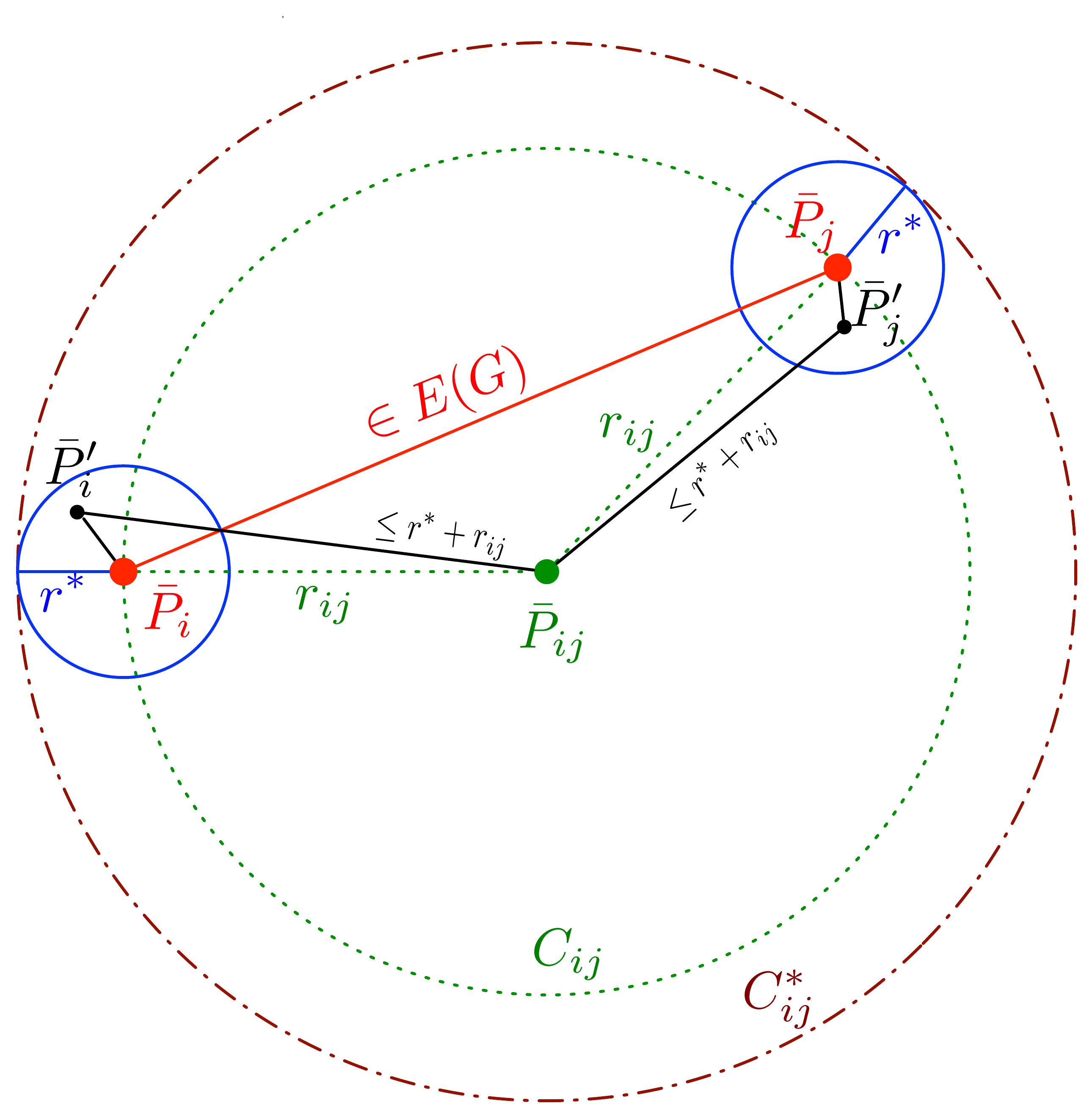}
\caption{Proving that the points $\bar P'_i$ and $\bar P'_j$ lie inside the disc $C^*_{ij}$ (Lemma~\ref{lem:edge-circle}).}
\label{fig1}
\end{figure}

\begin{figure}[t]
\centering
\includegraphics[scale=0.21]{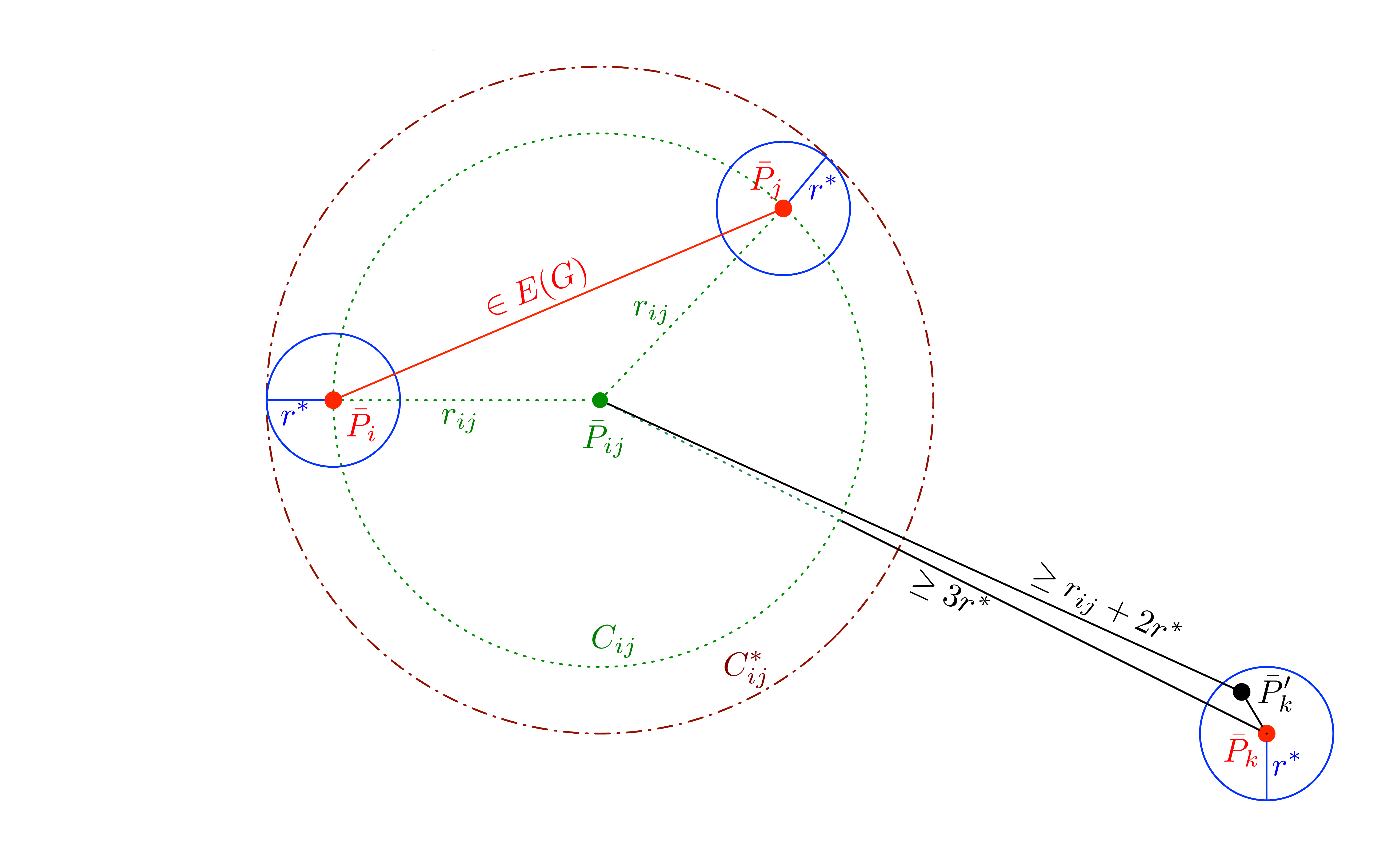}
\caption{Proving that the point $\bar P'_k$ lies outside the disc $C^*_{ij}$  (Lemma~\ref{lem:edge-circle}).}
\label{fig2}
\end{figure}

\begin{lemma}\label{lem:edge-circle}
$\mathscr{DT}(P)$ is isomorphic to $\mathscr{DT}(P')$ and the outer face of $\mathscr{DT}(P')$ consists of all the points corresponding to vertices in $V(f^*)$.
\end{lemma}
\begin{proof}
Let $\bar P_1,\bar P_2, \cdots, \bar P_{n_c}$ be the convex hull of $P$, where $n_c=|V(f^*)|$. We start by proving that $\bar P_1,\bar P_2, \cdots, \bar P_{n_c}$ forms the convex hull of $P'$. Towards this, we rely on the following property of convex hull. For a point set $Q$, an edge $(q,q')$ of the convex hull of $Q$, and a point $p \in Q \setminus \{q,q'\}$ let $H$ be the hyperplane containing $p$, which is defined by the line between $q$ and $q'$. Then for all $p' \in Q \setminus \{q,q'\}$, $p'$ lies in the hyperplane $H$. In fact, its converse holds as well, i.e., if for a pair of points all the other points are contained on the same hyperplane defined by them, then it forms an edge of the convex hull.

Consider an edge $(\bar P_i, \bar P_{j})$ of the convex hull of $P$, and the corresponding pair of points $(\bar P'_i, \bar P'_{j})$ in $P'$. We will show that $(\bar P'_i, \bar P'_{j})$ is an edge in the convex hull of $P'$. Assuming the contrary, suppose that $(\bar P'_i, \bar P'_{j})$ is not an edge of the convex hull of $P'$. This implies that there exist $\bar P'_k, \bar P'_s \in P'$ such that $\bar P'_k$ and $\bar P'_s$ lie on different sides (hyperplanes) defined by the (unique) line containing $\bar P'_i, \bar P'_{j}$. We will show that $\bar P_k$ and $\bar P_s$ lie on the opposite sides of the hyperplane defined by $\bar P_i$ and $\bar P_j$, thus arriving at a contradiction that $(\bar P_i, \bar P_{j})$ is an edge of the convex hull of $P$. Let $L$ be the line containing $\bar P_i, \bar P_j$. Consider the lines $L_U$ and $L_D$ which are at a distance $3r^*/2$ from $L$, and are parallel to $L$, but are on different sides of $L$. We note that $\bar P_k, \bar P_i, \bar P_j$ and $\bar P_s, \bar P_i, \bar P_j$ are not collinear, respectively. This follows from the fact that $\displaystyle{d^*_A} \geq 3r^*>0$. Note that $\bar P'_i$ and $\bar P'_j$ lie in the region between the lines $L_U$ and $L_D$ (and not on $L_U$ or $L_D$). This follows from the fact that $\textsf{dist}(\bar P_t, \bar P'_t) \leq r^*$, for all $t \in [n]$. Since $\displaystyle{d^*_A}\geq 3r^*$, the points $\bar P_k$ and $\bar P_s$ do not lie on $L_U$ or $L_D$. Also, they do not lie on the region between $L_U$ and $L_D$. This together with the fact that $\textsf{dist}(\bar P_k, \bar P'_k) \leq r^*$ implies that $\bar P'_k$ does not lie on the lines $L_U$ or $L_D$, or in the region between them. The symmetric argument holds for $\bar P'_s$. Without loss of generality, assume that $\bar P'_k$ lie on the same side of the line $L$ as line $L_D$. By our assumption that $\bar P'_k$ and $\bar P'_s$ lie on different sides (hyperplanes) defined by the line containing $\bar P'_i, \bar P'_{j}$, we deduce $\bar P'_s$ must lie on the side of $L$ containing the line $L_U$. Furthermore, we assume that both $\bar P_k, \bar P_s$ lie on the same side of $L$ as the line $L_D$ (the other case is symmetric). But this implies that $\textsf{dist}(\bar P_s, \bar P'_s) \geq 2 r^*$, a contradiction. Therefore, we have that $(\bar P'_i, \bar P'_j)$ is an edge of the convex hull of $P'$. Note that this implies that $\bar P_1,\bar P_2, \cdots, \bar P_{n_c}$ forms the convex hull of $P'$. 

Next, we elucidate the edge relationships between $\mathscr{DT}(P)$ and $\mathscr{DT}(P')$. To this end, consider an edge $(\bar P_i, \bar P_j) \in E(\mathscr{DT}(P))$, where $i,j \in [n]$. Moreover, consider the disc $C^*_{ij}$ co-centric to $C_{ij}$ with radius $r^*_{ij}=r_{ij}+r^*$, where $r_{ij}$ is the radius of $C_{ij}$ (see Fig.~\ref{fig1}). Our objective is to prove that $C^*_{ij}$ contains both $\bar P'_i$ and $\bar P'_j$, and excludes all other points in $P'$. This would imply that $(\bar P'_i,\bar P'_j) \in \mathscr{DT}(P')$ as desired. First, let us argue that $C^*_{ij}$ contains both $\bar P'_i$ and $\bar P'_j$. For this purpose, consider the triangle formed by the points $\bar P_{ij}, \bar P_i$ and $\bar P'_i$. From the triangle inequality it follows that $\textsf{dist}(\bar P'_i, \bar P_{ij}) \leq \textsf{dist}(\bar P_i,\bar P_{ij}) + \textsf{dist}(\bar P'_i,\bar P_i)$. However, $\textsf{dist}(\bar P'_i,\bar P_i) \leq r^*$, and therefore $\textsf{dist}(\bar P'_i, \bar P_{ij}) \leq r_{ij}+r^*$. Hence, it follows that $C^*_{ij}$ contains the point $\bar P'_{i}$. Symmetrically, we have that $C^*_{ij}$ contains the point $\bar P'_{j}$.

We now argue that for $k \in [n] \setminus \{i,j\}$, the point $\bar P'_k$ lies outside the disc $C^*_{ij}$. Note that $\textsf{dist}(C_{ij}, \bar P_k) \geq 3r^*$ (since $d^*_C \geq 3r^*$). In particular, $\textsf{dist}(\bar P_{ij}, \bar P_k) \geq r_{ij}+3r^*$ (see Fig.~\ref{fig2}). This, together with the triangle inequality in the context of the triangle formed by the points $\bar P_{ij}, \bar P_k$ and $\bar P'_k$, implies that $\textsf{dist}(\bar P_{ij}, \bar P'_k) \geq r_{ij}+ 2r^*$. Hence, $\bar P'_k$ indeed lies outside the disc $C^*_{ij}$. Overall, we conclude that $(\bar P'_i, \bar P'_j) \in E(\mathscr{DT}(P'))$. 

Note that both $\mathscr{DT}(P)$ and $\mathscr{DT}(P')$ are on $n$ points, and both have $|V(f^*)|$ vertices on the outer face. Further, we have shown that $|\mathscr{DT}(P)| \leq |\mathscr{DT}(P')|$. Then, from Theorem 9.1 in~\cite{comp-geom-berg} it follows that $|E(\mathscr{DT}(P'))| =  |E(\mathscr{DT}(P))|$. In other words, non-edges are also preserved. This concludes the proof. 
  \end{proof}

\begin{theorem}\label{lem:exists-circle}
Let $G$ be a \triangulation\ on $n$ vertices with $f^*$ as its outer face, realizable as a \DT\ where the points corresponding to vertices of $f^*$ lie on the outer face. Moreover, let ${\cal Q}$ be a solution of $\textsf{Const}(G)$ and $\beta \in \mathbb{R}^+$. Then, there is a solution ${\cal P}$ of $\textsf{Const}(G)$, assigning a set of points $P \subseteq \mathbb{R}^2$ to vertices of $G$, such that for each $v_i \in V(G)$, there exists a disc $C_i$ with centre $\bar P_i$ and radius at least $\beta$ for which the following condition holds. For any $P'=\{\bar P'_i \mid \bar P'_i \in C_i, i \in [n]\}$, it holds that $\mathscr{DT}(P')$ is isomorphic to $\mathscr{DT}(P)$, and the points corresponding to vertices of $f^*$ lie on the outer face of $\mathscr{DT}(P')$.
\end{theorem}
\begin{proof}
The proof of theorem follows directly from the construction of $r^*$, the discs $C_i$ for $i\in [n]$, and Lemma~\ref{lem:edge-circle}.
  \end{proof}

%% file: integer-coordinates-dt.tex
\section{\delaunaytr: Integer Coordinates} \label{sec:int-coordinates}

In this section, we prove our main theorem:

\begin{theorem}\label{thm:main}
Given a triangulation $G$ on $n$ vertices, in time $n^{\OO(n)}$ we can either output a point set $P \subseteq \mathbb{Z}^2$ such that $G$ is isomorphic to $\mathscr{DT}(P)$, or correctly conclude that $G$ is not Delaunay realizable.
\end{theorem}

\subparagraph{The Outer Face of the Output.} First, we explain how to identify the outer face $f^*$ of the output (in case the output should not be \no). For this purpose, let $f_{\textsf{out}}$ denote the outer face of $G$ (according to the embedding of the triangulation $G$, given as the input). Recall our assumption that $n \geq 4$. Let us first consider the case where $G$ is not a maximal planar graph, i.e., $f_{\textsf{out}}$ consists of at least four vertices. Suppose that the output is not \no. Then, for any point set $P \subseteq \mathbb{R}^2$ that realizes $G$ as a \DT, it holds that the points corresponding to the vertices of $f_{\textsf{out}}$ form the outer face of $\mathscr{DT}(P)$. Thus, in this case, we simply set $f^*=f_{\textsf{out}}$. Next, consider the case where $G$ is a maximal planar graph. Again, suppose that the output is not \no. Then, for a point set $P \subseteq \mathbb{R}^2$ that realizes $G$ as a \DT, the outer face of $\mathscr{DT}(P)$ need not be the same as $f_{\textsf{out}}$. To handle this case, we ``guess'' the outer face of the output (if it is not \no). More precisely, we examine each face $f$ of $G$ separately, and attempt to solve the ``integral version'' of \connCDT{} with $f^*$ set to $f$, and where $G$ is embedded with $f^*$, rather than $f_{\textsf{out}}$, as its outer face. Here, note that a maximal planar graph is $3$-connected~\cite{Whitney1992}, and therefore, by Proposition~\ref{convex-drawing-plane-graph}, we can indeed compute an embedding of $G$ with $f^*$ as the outer~face.

The number of iterations is  bounded by $\OO(n)$ (since the number of faces of $G$ is bounded by $\OO(n)$). Thus, from now on, we may assume that  we seek only  Delaunay realizations of $G$ where the outer face is the same as the outer face of $G$ (that we denote by $f^*$). 

%\medskip
%\noindent
\subparagraph{Sieving \no-Instances.} We compute the set $\textsf{Const}(G)$ as described in Section~\ref{sec:3-conn-DR}. From Theorem~\ref{thm:output-polynmials-dt}, we know that $G$ is realizable as a \DT\ with the points corresponding to $f^*$ on the outer face if and only if $\textsf{Const}(G)$ is satisfiable. Using Proposition~\ref{prop:solving-equations}, we check whether $\textsf{Const}(G)$ is satisfiable, and if the answer is negative, then we return \no.
Thus, we next focus on the following problem.

\defproblemout{\delaunaytrIntFull{} (\delaunaytrInt)}{A \triangulation\ $G$ with outer face $f^*$ that is realizable as a \DT\ with outer face $f^*$.}{A point set $P \subseteq \mathbb{Z}^2$ realizing $G$ as a \DT{} with outer face~$f^*$.} 

Similarly, we define the intermediate \delaunaytrRatFull{} (\delaunaytrRat) problem---here, however, $P \subseteq \mathbb{Q}^2$ rather than $\mathbb{Z}^2$. To prove Theorem~\ref{thm:main}, it is sufficient to prove the following result, which is the objective of the rest of this paper.

\begin{lemma}\label{thm:mainLemma}
\delaunaytrInt{} is solvable in time $n^{\OO(n)}$.
\end{lemma}

In what follows, we crucially rely on the fact that by Theorem~\ref{lem:exists-circle}, for all $\beta \in \mathbb{R}^+$, there is a solution ${\cal P}$ of $\textsf{Const}(G)$ that assigns a set of points $P \subseteq \mathbb{R}^2$ to the vertices of $G$, such that for each $v_i \in V(G)$, there exists a disc $C_i$ with radius at least $\beta$, satisfying the following condition: For any $P'=\{\bar P'_i \mid P'_i \in C_i, i \in [n]\}$, it holds that $\mathscr{DT}(P')$ is isomorphic to $G$ with points corresponding to the vertices in $f^*$ on the outer face (in the same order as in $f^*$).

As it would be cleaner to proceed while working with squares, we need the next observation.

\begin{observation}\label{obs:circle-int-coordinate}
Every disc $C$ with radius at least $2$ contains a square of side length at least $2$ and with the same centre. %point $(x,y) \in \mathbb{Z}^2$ inside or on its boundary.
\end{observation}

We next extend $\textsf{Const}(G)$  to a set $\textsf{ConstSqu}(G)$, which explicitly ensures that there exists a {\em square} around each point in the solution such that the point can be replaced by any point in the square. Thus, rather than discs of radius $2$ (whose existence, in some solution, is proven by choosing $\beta=2$), we consider squares with side length $2$ given by Observation \ref{obs:circle-int-coordinate}, and force our constraints to be satisfied at the corner points of the squares. For this purpose, for each $v_i \in V(G)$, apart from adding constraints for the point $P_i =(X_i,Y_i)$ (which can be regarded as a disc of radius $0$ in the previous setting), we also have constraints for the corner points of the square of side length 2 whose centre is $P_i$. For technical reasons, we also add constraints for the intersection points of perpendicular bisectors. For any constraint where $P_i$ appears, we make copies for the points $P_i=(X_i,Y_i)$, $P^1_i=(X_i-1,Y_i -1)$, $P^2_i=(X_i-1,Y_i +1)$, $P^3_i=(X_i+1,Y_i -1),P^4_i=(X_i+1,Y_i +1), P^5_i=(X_i-1,Y_i), P^6_i=(X_i,Y_i+1), P^7_i=(X_i+1,Y_i), P^8_i=(X_i,Y_i-1)$.\footnote{We remark that we do not create new variables for the corresponding $x$- and $y$-coordinates for points $P^\ell_i$, for $i\in [n]$.} % We are now ready to present the set of polynomials, $\textsf{ConstDisc}(G)$, which we would like to satisfy. % These set of inequalities extends the one in Section~\ref{sec:3-conn-DR} according to the discussion above.%Further, we will prove that a ``good'' approximation to $\textsf{ConstEqDisc}(G)$ is indeed an exact solution to our problem.

%\medskip
%\noindent
%{\bf Inequalities that ensure the outer face forms the convex hull.} 
\subparagraph{Inequalities that ensure the outer face forms the convex hull.} 
We generate the set of constraints that ensure the points corresponding to vertices in $V(f^*)$ form a convex hull of the output point set. Let $C^*$ be the cycle of the outer face $f^*$.  Whenever we say consecutive vertices in $C^*$, we always follow clockwise direction. For three consecutive vertex $v_i,v_j$ and $v_k$ in $C^*$, for every $Z_i \in \{P_i\} \cup \{P^\ell_i \mid \ell \in [8]\}$, $Z_j \in \{P_j\} \cup \{P^\ell_j \mid \ell \in [8]\}$ and $Z_k \in \{P_k\} \cup \{P^\ell_k \mid \ell \in [8]\}$, we add the inequality $\textsf{Con}(Z_i,Z_j,Z_k)>0$. This ensures that the points corresponding to vertices in $V(f^*)$ are in convex position in any output point set. Further, we want all the points which correspond to the vertices in $V(G) \setminus V(f^*)$ to be in the interior of the convex hull formed by the points corresponding to vertices in $V(f^*)$. To achieve this, for each pair of vertices $v_i,v_j$ that are consecutive vertices of $C^*$, $v_k \in V(G) \setminus \{v_i,v_j\}$, $Z_i \in \{P_i\} \cup \{P^\ell_i \mid \ell \in [8]\}$, $Z_j \in \{P_j\} \cup \{P^\ell_j \mid \ell \in [8]\}$ and $Z_k \in \{P_k\} \cup \{P^\ell_k \mid \ell \in [8]\}$, we add $\textsf{Con}(Z_i,Z_k,Z_j)<0$. We call the above set of polynomial constraints $\textsf{ConSqu}(G)$.

%\medskip
%\noindent
\subparagraph{Inequalities that guarantee existence of edges.}
For each edge $(v_i,v_j) \in E(G)$, we add three new variables, namely $X_{ij}, Y_{ij}$ and $r_{ij}$. These newly added variables will correspond to the centre and radius of a disc that realizes the edge $(v_i,v_j)$. There might exist many such discs, but we are interested in only one such disc. In particular, $(X_{ij}, Y_{ij})$ corresponds to centre of one such discs, say $C_{ij}$, with radius $r_{ij}$, containing all the points in $\{P_i\} \cup \{P^\ell_i \mid \ell \in [8]\}$ and $\{ P_j\} \cup \{P^\ell_j \mid \ell \in [8]\}$ but none of the points in $\{P_k \mid k \in [n] \setminus \{i,j\}\} \cup \{P^\ell_k \mid \ell \in [8], k \in [n] \setminus \{i,j\}\}$. Towards this, we add a set of inequalities for each edge $(v_i,v_j) \in E(G)$, which we will denote by $\textsf{DisSqu}(v_i,v_j)$. For each $Z \in \{P_i, P_j\} \cup \{P^\ell_i, P^\ell_j \mid \ell \in [8]\}$, we add the following inequalities to $\textsf{DisSqu}(v_i,v_j)$, ensuring that $C_{ij}$ contains $Z=(Z_X,Z_Y)$.

$$Z^2_X + X^2_{ij}-2Z_XX_{ij}+Z^2_Y + Y^2_{ij} - 2Z_Y Y_{ij}-r^2_{ij} \leq 0.$$

Further, we want to ensure that for each $k \in [n] \setminus \{i,j\}$, $Z \in \{P_k\} \cup \{P^\ell_k \mid \ell \in [8]\}$ does not belong to $C_{ij}$. Hence, for each such $Z = (Z_X,Z_Y)$, the following must hold.

$$Z^2_X + X^2_{ij}-2Z_XX_{ij}+Z^2_Y + Y^2_{ij} - 2Z_Y Y_{ij} - r^2_{ij} > 0.$$

Hence, we add the above constraint to $\textsf{DisSqu}(v_i,v_j)$ for $k \in [n] \setminus \{i,j\}$. We denote $\displaystyle{\textsf{DisSqu}(G)=\bigcup_{(v_i,v_i) \in E(G)}\textsf{DisSqu}(v_i,v_j)}$. 

This completes the description of all the constraints we need. We let $\textsf{ConstSqu}(G)= \textsf{ConSqu}(G) \cup \textsf{DisSqu}(G)$. We let $n^*$ denote the number of variables appearing in $\textsf{ConstSqu}(G)$. Note that $n^*=\OO(n)$ and the number of constraints in $\textsf{ConstSqu}(G)$ is bounded by $\OO(n^2)$.%, and we would seek a solution ${\cal P}$ of $\textsf{ConstSqu}(G)$. 

\begin{theorem}\label{thm:output-polynmials-dt-integers}
Let $G$ be a triangulation on $n$ vertices with $f^*$ as the outer face. Then, in time $\OO(n^2)$ we can find a set of polynomial constraints $\textsf{ConstSqu}(G)$ such that $G$ is realizable as a \DT\ with $f^*$ as its outer face if and only if $\textsf{ConstSqu}(G)$ is satisfiable. Moreover, $\textsf{ConstSqu}(G)$ consists of $\OO(n^2)$ constraints and $\OO(n)$ variables, where each constraint is of degree 2 and with coefficients only from $\{-10,-9,\ldots,10\}$.
\end{theorem}
\begin{proof}
Follows from the construction of $\textsf{ConstSqu}(G)$, Lemma~\ref{lem:correctness-dr-fwd-3-conn}, and Theorems~\ref{lem:exists-circle}.
\end{proof}

%We will prove that a ``good'' approximation to $\textsf{ConstEqDisc}(G)$ is indeed an exact solution to our problem.
%We let $n^*$ denote the number of variables appearing in $\textsf{ConstDisc}(G)$. Note that the number of inequalities in $\textsf{ConstDisc}(G)$ is bounded by $\OO(n^2)$.% and we will be looking for ${\cal P} \in \textsf{Zer}(\textsf{ConstDiscEq}(G), \mathbb{R}^{n^*})$. %From the fact that $\textsf{Zer}(\textsf{ConstEq}(G)$, $\mathbb{R}^{n^*})$ is non-empty and Theorem~\ref{lem:exists-circle} we know that $\textsf{Zer}(\textsf{ConstDiscEq}(G), \mathbb{R}^{n^*})$ is non-empty.

Having proved Theorem~\ref{thm:output-polynmials-dt-integers}, we use Proposition~\ref{prop:solving-equations} to decide in time $n^{\OO(n)}$ if $\textsf{ConstSqu}(G)$ is satisfiable. Recall that if the answer is negative, then we returned \no. We compute a ``good'' approximate solution as we describe next. First, by Proposition~\ref{prop:solving-equations}, in time $n^{\OO(n)}$ we compute a (satisfiable) set $\mathfrak{C}$ of $n^*$ polynomial constraints, $\mathscr{C}_1,\mathscr{C}_2,\ldots,\mathscr{C}_{n^*}$, with coefficients in $\mathbb{Z}$, where for all $i\in[n]$, we have that $\mathscr{C}_i$ is an equality constraint on the variable indexed $i$ (only), and a solution of $\mathfrak{C}$ is also a solution of $\textsf{ConstSqu}(G)$.
Next, we would like to find a ``good'' rational approximation to the solution of $\mathfrak{C}$. Later we will prove that such an approximate solution is actually an {\em exact} solution to our problem.

For $\delta >0$, a $\delta$ rational approximate solution $\cal S$ for a set of polynomial equality constraints is an assignment to the variables, for which there exists a solution ${\cal S}^*$, such that for any variable $X$, the (absolute) difference between the assignment to $X$ by $\cal S$ and the assignment to $X$ by ${\cal S}^*$ is at most $\delta$.\footnote{We note that $\cal S$ may not be a solution in the sense that it may not satisfy all constraints (but it is close to some solution that satisfies all of them).} We follow the approach of Arora et al.~\cite{Arora:2012} to find a $\delta$ rational approximation to a solution for a set of polynomial equality constraints with $\delta=1/2$. This approach states that using Renegar's algorithm~\cite{Renegar:1992} together with binary search, with search range bound given by Grigor'ev and Vorobjov~\cite{Grigorev88}, we can find a rational approximation to a solution of a set of polynomial equality constraints with accuracy up to $\delta$ in time $(\tau +n'+m'^{n'} +\log (1/ \delta))^{\OO(1)}$ where $\tau$ is the maximum bitsize of a coefficient, $n'$ is the number of variables and $m'$ is the number of constraints. In this manner, we obtain in time $n^{\OO(n)}$ a rational approximation ${\cal S}$ to the solution of $\mathfrak{C}$ with accuracy $1/2$. By Theorem~\ref{thm:output-polynmials-dt-integers}, ${\cal S}$ is also a rational approximation to a solution of $\textsf{ConstSqu}(G)$ with accuracy $1/2$. We let $\bar P_{S_i} = (\bar X_{S_i}, \bar Y_{S_i})$ denote the value that ${\cal S}$ assigns to $(X_i,Y_i)$ (corresponding to the vertex $v_i \in V(G)$). Further we let $P_S = \{\bar P_{S_i} \mid i \in [n]\}$. In the following lemma, we analyze $\mathscr{DT}(P_S)$.

\begin{lemma}
\label{lem:approx-sol}
The triangulation $G$ is isomorphic to $\mathscr{DT}(P_S)$ where  points corresponding to vertices in $f^*$ form the outer face (in that order). Here, $P_S$ is the point set described above.
\end{lemma}
\begin{proof}
Note that ${\cal S}$ is a rational approximation to a solution ${\cal P}$ of $\textsf{ConstSqu}(G)$ with accuracy $1/2$. We let $\bar P_{i} = (\bar X_{i}, \bar Y_{i})$ denote the value ${\cal P}$ assigns to $(X_i,Y_i)$ corresponding to the vertex $v_i \in V(G)$, and $P = \{\bar P_{i} \mid i \in [n]\}$. For an edge $(\bar P_i, \bar P_j) \in E(\mathscr{DT}(P))$, let $\bar P_{ij}$, $r_{ij}$ be the centre and radius, respectively, of the disc assigned by $\cal P$, realizing $(\bar P_i, \bar P_j)$. For $i \in [n]$, we denote $\bar P_i=(\bar X_i,\bar Y_i)$, $\bar P^1_i=(\bar X_i-1,\bar Y_i -1)$, $\bar P^2_i=(\bar X_i-1,\bar Y_i +1)$, $\bar P^3_i=(\bar X_i+1,\bar Y_i -1), \bar P^4_i=(\bar X_i+1,\bar Y_i +1), \bar P^5_i=(X_i-1,Y_i),  \bar P^6_i=(X_i,Y_i+1),  \bar P^7_i=(X_i+1,Y_i)$ and $\bar P^8_i=(X_i,Y_i-1)$. For $i \in [n]$, we let $\bar P_{S_i} = (\bar X_{S_i}, \bar Y_{S_i})$ denote the value that ${\cal S}$ assigns for the point $P_i \in {\cal P}$, and $P_S= \{\bar P_{S_i} \mid i \in [n]\}$. We let ${\cal R}_i$ be the axis-parallel square with side bisectors intersecting at ${\bar P}_i$ of side length $2$. Notice that ${\cal R}_i$ contains $\{\bar P^\ell_i \mid \ell \in [8]\}$ on its boundary. Furthermore, $\bar P^1_i, \bar P^2_i, \bar P^3_i$ and $\bar P^4_i$ are its corner points. Since ${\cal S}$ is an $1/2$-accurate approximation to ${\cal P}$, this implies that for each $i \in [n]$, $|\bar X_{S_i} - \bar X_i | \leq 1/2$ and $|\bar Y_{S_i} - \bar Y_i | \leq 1/2$, and hence $\bar P_{S_i}$ is strictly contained inside the square ${\cal R}_i$.  

Consider an edge $(\bar P_i, \bar P_j) \in E(\mathscr{DT}(P))$. Let us recall that from the construction of $\textsf{ConstSqu}(G)$ and since ${\cal P}$ satisfies $\textsf{ConstSqu}(G)$, we have that $C_{ij}$ contains all the points in $\{\bar P_i, \bar P_j\} \cup \{ \bar P^\ell_i, \bar P^\ell_j \mid \ell \in [8]\}$,  and it contains none of the points in $\{\bar P_k \mid k \in [n] \setminus \{i,j\}\} \cup \{\bar P^\ell_k \mid \ell \in [8], k \in [n] \setminus \{i,j\}\}$. By properties of convex sets, this implies that $C_{ij}$ contains the points $\bar P_{S_i}$ and $\bar P_{S_j}$. Now, we prove that $C_{ij}$ does not contain any point $\bar P_{S_k}$ where $k \in [n] \setminus \{i,j\}$. Targeting a contradiction, suppose for some $k \in [n] \setminus \{i,j\}$, $C_{ij}$ contains $\bar P_{S_k}$. Since $\bar P_{S_k}$ is contained strictly inside ${\cal R}_k$ and pseudo discs can intersects in either one point or the intersection contains exactly two points of the boundary, this implies that $C_{ij}$ intersects ${\cal R}_k$ at two points on the boundary. Since $C_{ij}$ contains none of the points in $\{\bar P^\ell_k \mid \ell \in [8]\}$, we have that the two intersecting points on the boundary must lie on exactly one line segment between two consecutive points in $\{\bar P_k^\ell \mid \ell \in [8]\}$ of one of the sides of ${\cal R}_k$---without loss of generality, say they lie on the line segment $\overline{\bar P^1_k\bar P^5_k}$.

Consider the line segment $\bar L_{Ck}$ joining the centre $\bar P_{ij}$ of the disc $C_{ij}$ and $\bar P_{S_k}$, and let $Z$ be the point of its intersection with $\overline{\bar P^1_k\bar P^5_k}$. One of the line segments $\overline{\bar P^1_kZ}$ or $\overline{Z \bar P^5_k}$ has length at most $1/2$, and the other has length at least $1/2$---say $\overline{\bar P^1_kZ}$ has length at most $1/2$ (our arguments also hold for the case where $\overline{Z \bar{P}^5_k}$ has length at most $1/2$). Let $w$ be the length of the line segment $\overline{\bar P_{ij}Z}$. Since the difference between the $x$-coordinates of $Z$ and $\bar{P}_{k}$ is $1$ while the difference between the $x$-coordinates of $\bar P_{S_k}$ and $\bar{P}_{k}$ is at most $1/2$, we have that the length of the line segment $\overline{Z \bar P_{S_k}}$ is at least $1/2$. We deduce that the length of $\bar L_{Ck}$ is at least $w+1/2$.
Thus, the radius of $C_{ij}$ is at least $w+1/2$. By the triangle inequality, it follows that the length of line segment $\overline{\bar P_{ij}\bar P^1_k}$ is at most $w+1/2$. This contradicts the fact that $C_{ij}$ does not contain $\bar P^1_k$. Hence, it follows that $C_{ij}$ does not contain any point in $\{\bar P_{S_k} \mid k \in [n] \setminus \{i,j\}\}$. Therefore, $(\bar P_{S_i}, \bar P_{S_j}) \in E(\mathscr{DT}(P_S))$. 

We now argue that the points corresponding to the vertices in $V(f^*)$ form the convex hull of $P_S$ in the order in which they appear in the cycle bounding $f^*$. Consider two consecutive vertices, $v_i$ and $v_j$, in the cycle of $f^*$. Since ${\cal P}$ satisfies $\textsf{ConSqu}(G)$, we have that for any $Z_i \in {\cal R}_i$ and $Z_j \in {\cal R}_j$, there exists a line $L_{ij}$ such that all the points in $\{P^\ell_k \mid \ell \in [8], k \in [n] \setminus \{i,j\}\}$ lie on one side (on side of the half space) and all the points in $\{P^\ell_i,P^\ell_j \mid \ell \in [8]\}$ lie on the opposite side the line $L_{ij}$. Here, for ensuring that these points do not lie on the line $L_{ij}$ we rely on the definition of $\textsf{ConSqu}(G)$ given in Section~\ref{sec:int-coordinates}, which ensures strict inequality ($<$ or $>$).

%one of the half spaces (sides) it defines contains all the points in $\{P^\ell_k \mid \ell \in [8], k \in [n] \setminus \{i,j\}\}$, and the other half space contains all the points in $\{P^\ell_i,P^\ell_j \mid \ell \in [8]\}$. 

Note that the points $\{\bar P^\ell_i \mid \ell \in [8]\}$ form a convex hull of all the points contained in the square ${\cal R}_i$, and the points $\{\bar P^\ell_j \mid \ell \in [8]\}$ form a convex hull of all the points contained in the square ${\cal R}_j$. But then $L_{ij}$ is a line such that $\bar P_{S_i}$ and $\bar P_{S_j}$ are contained in one of the half spaces defined by $L_{ij}$, and all the points in $\{\bar P_{S_k} \mid k \in [n] \setminus \{i,j\}\}$ are contained in the other half space. Hence it follows that points corresponding to vertices in $(\bar P_{S_i} , \bar P_{S_j})$ is an edge of  the convex hull of $P_S$. This implies that for the convex hull say, $(\bar P_1, \cdots \bar P_{n_c})$ of $P$ we have that $(\bar P_{S_1}, \cdots \bar P_{S_{n_c}})$ is the convex hull of $P_S$. 

Notice that we have shown that $|E(\mathscr{DT}(P))| \leq |E(\mathscr{DT}(P_S))|$, and $(\mathscr{DT}(P))$ and $(\mathscr{DT}(P_S))$ have same number of vertices on the outer face. Hence, by Theorem 9.1~\cite{comp-geom-berg}, it follows that $|E(\mathscr{DT}(P))| = |E(\mathscr{DT}(P_S))|$. This concludes the proof.
\end{proof}
Towards the proof of Lemma~\ref{thm:mainLemma}, we first consider our intermediate problem.

\begin{lemma}\label{thm:main-thm-dt}
\delaunaytrRat{} is solvable in time $n^{\OO(n)}$.
\end{lemma}
\begin{proof}
Our algorithm first computes the set of polynomial constraints $\textsf{ConstSqu}(G)$ in time $\OO(n^2)$. Then, it computes a $1/2$ accurate approximate solution for $\textsf{ConstSqu}(G)$ by using the approach of Arora et al.~\cite{Arora:2012} in time $n^{\OO(n)}$. In Lemma~\ref{lem:approx-sol}, we have shown that such an approximate solution is an exact solution. This concludes the proof.
\end{proof}
%\begin{proof}
%Our algorithm first computes the set of polynomial constraints $\textsf{ConstSqu}(G)$ in time $\OO(n^2)$. Then, it computes a $1/2$ accurate approximate solution for $\textsf{ConstSqu}(G)$ by using the approach of Arora et al.~\cite{Arora:2012} in time $n^{\OO(n)}$. In Lemma~\ref{lem:approx-sol}, we have shown that such an approximate solution is an exact solution. This concludes the proof.
%Theorem~\ref{thm:integer-theorem-dt} implies that if $G$ is realizable as \DT\ with $f^*$ as the outer face then there exists ${\cal P} \in \textsf{Zer}(\textsf{ConstEq}(G)$, $\mathbb{Z}^{n^*})$ assigning the set of points $P \subseteq \mathbb{Z}^2$ to vertices of $G$ such that $\mathscr{DT}(P)$ is isomorphic to $G$ with points corresponding to vertices in $V(f^*)$ on its outer face. In this case using second part of the Proposition~\ref{prop:solving-equations} and by setting $C=\mathbb{Z}$ in time $2^{\OO(1)}$ we obtain ${\cal P} \in \textsf{Zer}(\textsf{ConstEq}(G))$ assigning a point $\bar P_i$ to the vertex $v_i \in V(G)$ such that $G$ is isomorphic to $\mathscr{DT}(P)$. This concludes the proof. 
%\end{proof}

%Armed with Lemma~\ref{thm:main-thm-dt}, we can prove Lemma~\ref{thm:mainLemma}, and thus conclude the correctness of Theorem~\ref{thm:main}. The proof Lemma~\ref{thm:mainLemma} can be found in the appendix. 

Finally, we are ready to prove Lemma~\ref{thm:mainLemma}, and thus conclude the correctness of Theorem~\ref{thm:main}.

\begin{proof}[Proof of Lemma~\ref{thm:mainLemma}]%{Proof of Lemma~\ref{thm:mainLemma}}
%\subparagraph*{Proof of Lemma~\ref{thm:mainLemma}.}
We use the algorithm given by Lemma~\ref{thm:main-thm-dt} to output a point set $P \subseteq \mathbb{Q}^2$ in time $n^{\OO(n)}$ such that $G$ is isomorphic to $\mathscr{DT}(P)$ and the points corresponding to vertices in $V(f^*)$ lie on the outer face of $\mathscr{DT}(P)$ in the order in which they appear in the cycle of $f^*$. We denote by $\bar P_i= (\bar X_i, \bar Y_i)$ the value $P$ assigns to the vertex $v_i \in V(G)$. For $i \in [n]$, since $\bar X_i, \bar Y_i \in \mathbb{Q}$, we let the representation be $\bar X_i= \bar X^a_i/ \bar X^b_i$ and $\bar Y_i= \bar Y^a_i/ \bar Y^b_i$, where $\bar X^a_i, \bar X^b_i, \bar Y^a_i, \bar Y^b_i \in \mathbb{Z}$. For each edge $(\bar P_i, \bar P_j) \in E(\mathscr{DT}(P))$, there exists a disc $C_{ij}$ with a centre, say $\bar P_{ij}$, containing only $\bar P_i$ and $\bar P_j$ from $P$. These assignments satisfy the constraints $\textsf{Con}(G)$ and $\textsf{Dis}(G)$ presented in Section~\ref{sec:3-conn-DR}. It thus follows that ${\cal P}$ satisfies $\textsf{Const}(G)$. From Observation~\ref{obs:multiply-by-const-soln} it follows that for any $\alpha \in \mathbb{R}^+$, we have that $\alpha {\cal P}$ satisfies $\textsf{Const}(G)$. We let $\beta= \Pi_{i \in [n]} \bar X^b_i\bar Y^b_i$. But then $\beta {\cal P}$ satisfies $\textsf{Const}(G)$, and hence $\beta P =\{(\beta \bar X_i,\beta \bar Y_i \mid i \in [n]\}$ is a point set such that $G$ is isomorphic to $\mathscr{DT}(\beta P)$ where the points corresponding to vertices in $V(f^*)$ lie on the outer face of $\mathscr{DT}(\beta P)$. Therefore, we output a correct point set, $\beta P$, with only integer coordinates. This concludes the proof.
  \end{proof}

%% file: main-dt.bbl
\begin{thebibliography}{10}

\bibitem{DBLP:journals/dcg/AdiprasitoPT15}
Karim~A. Adiprasito, Arnau Padrol, and Louis Theran.
\newblock Universality theorems for inscribed polytopes and delaunay
  triangulations.
\newblock {\em Discrete {\&} Computational Geometry}, 54(2):412--431, 2015.

\bibitem{AlamRS11}
Md.~Ashraful Alam, Igor Rivin, and Ileana Streinu.
\newblock Outerplanar graphs and {D}elaunay triangulations.
\newblock In {\em Proceedings of the 23rd Annual Canadian Conference on
  Computational (CCCG)}, 2011.

\bibitem{Arora:2012}
Sanjeev Arora, Rong Ge, Ravindran Kannan, and Ankur Moitra.
\newblock Computing a nonnegative matrix factorization -- provably.
\newblock In {\em Proceedings of the 44th Annual ACM Symposium on Theory of
  Computing}, STOC, pages 145--162, 2012.

\bibitem{Artin:1998}
M.~Artin.
\newblock {\em Algebra}.
\newblock Pearson Prentice Hall, 2011.

\bibitem{book-real-algebraic-geometry-basu}
Saugata Basu, Richard Pollack, and Marie-Fran\c{c}oise Roy.
\newblock {\em Algorithms in Real Algebraic Geometry (Algorithms and
  Computation in Mathematics)}.
\newblock Springer-Verlag New York, Inc., Secaucus, NJ, USA, 2006.

\bibitem{comp-geom-berg}
Mark~de Berg, Otfried Cheong, Marc~van Kreveld, and Mark Overmars.
\newblock {\em Computational Geometry: Algorithms and Applications}.
\newblock Springer-Verlag TELOS, 3rd ed. edition, 2008.

\bibitem{DBLP:journals/acta/ChibaON85}
Norishige Chiba, Kazunori Onoguchi, and Takao Nishizeki.
\newblock Drawing plane graphs nicely.
\newblock {\em Acta Inf.}, 22(2):187--201, 1985.

\bibitem{Clarkson:1989}
K.~L. Clarkson and P.~W. Shor.
\newblock Applications of random sampling in computational geometry, {II}.
\newblock {\em Discrete Computational Geometry}, 4:387--421, 1989.

\bibitem{book-algorithms-cormen}
Thomas~H. Cormen, Charles~E. Leiserson, Ronald~L. Rivest, and Clifford Stein.
\newblock {\em Introduction to Algorithms {(3.} ed.)}.
\newblock {MIT} Press, 2009.

\bibitem{DBLP:conf/stoc/FraysseixPP88}
Hubert de~Fraysseix, J{\'{a}}nos Pach, and Richard Pollack.
\newblock Small sets supporting f{\'{a}}ry embeddings of planar graphs.
\newblock In {\em Proceedings of the 20th Annual {ACM} Symposium on Theory of
  Computing (STOC)}, pages 426--433, 1988.

\bibitem{DBLP:journals/siamdm/BattistaV96}
Giuseppe {Di Battista} and Luca Vismara.
\newblock Angles of planar triangular graphs.
\newblock {\em SIAM Journal on Discrete Mathematics}, 9(3):349--359, 1996.

\bibitem{diestel-book}
Reinhard Diestel.
\newblock {\em Graph Theory, 4th Edition}, volume 173 of {\em Graduate texts in
  mathematics}.
\newblock Springer, 2012.

\bibitem{Dillencourt:1987:TDT}
M.~Dillencourt.
\newblock Toughness and {D}elaunay triangulations.
\newblock In {\em Proceedings of the Third Annual Symposium on Computational
  Geometry}, SoCG, pages 186--194, 1987.

\bibitem{Dillencourt:1990}
Michael.~B. Dillencourt.
\newblock Realizability of {D}elaunay triangulations.
\newblock {\em Information Processing Letters}, 33:283--287, 1990.

\bibitem{Dillencourt-DT}
Michael~B. Dillencourt and Warren~D. Smith.
\newblock Graph-theoretical conditions for inscribability and {D}elaunay
  realizability.
\newblock {\em Discrete Mathematics}, 161(1-3):63--77, 1996.

\bibitem{Grigorev88}
D.~Yu. Grigor'ev and N.~N. Vorobjov, Jr.
\newblock Solving systems of polynomial inequalities in subexponential time.
\newblock {\em Journal of Symbolic Computation}, 5:37--64, 1988.

\bibitem{Guibas1992}
Leonidas~J. Guibas, Donald~E. Knuth, and Micha Sharir.
\newblock Randomized incremental construction of {D}elaunay and {V}oronoi
  diagrams.
\newblock {\em Algorithmica}, 7(1):381--413, 1992.

\bibitem{hiroshima2000}
Tetsuya Hiroshima, Yuichiro Miyamoto, and Kokichi Sugihara.
\newblock Another proof of polynomial-time recognizability of {D}elaunay
  graphs.
\newblock {\em IEICE Transactions on Fundamentals of Electronics,
  Communications and Computer Sciences}, 83:627--638, 2000.

\bibitem{hodgson1992char}
Craig~D Hodgson, Igor Rivin, and Warren~D Smith.
\newblock A characterization of convex hyperbolic polyhedra and of convex
  polyhedra inscribed in the sphere.
\newblock {\em Bulletin of the American Mathematical Society}, 27:246--251,
  1992.

\bibitem{DBLP:journals/jct/KratochvilM94}
Jan Kratochv{\'{\i}}l and Ji{v{r}}{'{i}} Matou{v{s}}ek.
\newblock Intersection graphs of segments.
\newblock {\em J. Comb. Theory, Ser. {B}}, 62(2):289--315, 1994.

\bibitem{lambert1997}
Timothy Lambert.
\newblock An optimal algorithm for realizing a {D}elaunay triangulation.
\newblock {\em {Information Processing Letters}}, 62(5):245--250, 1997.

\bibitem{DBLP:journals/jct/McDiarmidM13}
Colin McDiarmid and Tobias M{\"{u}}ller.
\newblock Integer realizations of disk and segment graphs.
\newblock {\em Journal of Combinatorial Theory, Series B}, 103(1):114--143,
  2013.

\bibitem{DBLP:journals/siamdm/MullerLL13}
Tobias M{\"{u}}ller, Erik~Jan van Leeuwen, and Jan van Leeuwen.
\newblock Integer representations of convex polygon intersection graphs.
\newblock {\em {SIAM} J. Discrete Math.}, 27(1):205--231, 2013.

\bibitem{DBLP:journals/ieicet/NishizekiMR04}
Takao Nishizeki, Kazuyuki Miura, and Md.~Saidur Rahman.
\newblock Algorithms for drawing plane graphs.
\newblock {\em {IEICE} Transactions}, 87-D(2):281--289, 2004.

\bibitem{OISHI}
Yasuaki Oishi and Kokichi Sugihara.
\newblock Topology-oriented divide-and-conquer algorithm for {V}oronoi
  diagrams.
\newblock {\em Graphical Models and Image Processing}, 57:303 -- 314, 1995.

\bibitem{Okabe-comp-geom-book}
Atsuyuki Okabe, Barry Boots, and Kokichi Sugihara.
\newblock {\em Spatial Tessellations: Concepts and Applications of {Voronoi}
  Diagrams}.
\newblock John Wiley \& Sons, Inc., 1992.

\bibitem{opac-b1093656}
János Pach and Pankaj~K. Agarwal.
\newblock {\em {Combinatorial Geometry}}.
\newblock Wiley-Interscience series in discrete mathematics and optimization.
  Wiley, New York, 1995.

\bibitem{Renegar:1992}
James Renegar.
\newblock On the computational complexity and geometry of the first-order
  theory of the reals.
\newblock {\em Journal of Symbolic Computation}, 13:255--352, 1992.

\bibitem{rivin1994euclidean}
Igor Rivin.
\newblock Euclidean structures on simplicial surfaces and hyperbolic volume.
\newblock {\em Annals of Mathematics}, 139:553--580, 1994.

\bibitem{Sugihara:1994}
Kokichi Sugihara.
\newblock Simpler proof of a realizability theorem on {D}elaunay
  triangulations.
\newblock {\em {Information Processing Letters}}, 50:173--176, 1994.

\bibitem{sugihara1992const}
Kokichi Sugihara and Masao Iri.
\newblock Construction of the {V}oronoi diagram for one million generators in
  single-precision arithmetic.
\newblock {\em Proceedings of the IEEE}, 80:1471--1484, 1992.

\bibitem{sugihara1994robust}
Kokichi Sugihara and Masao Iri.
\newblock A robust topology-oriented incremental algorithm for {V}oronoi
  diagrams.
\newblock {\em International Journal of Computational Geometry \&
  Applications}, 4(02):179--228, 1994.

\bibitem{Tutte-draw-graphs}
William~Thomas Tutte.
\newblock How to draw a graph.
\newblock {\em Proceedings of the London Mathematical Society}, 3(1):743--767,
  1963.

\bibitem{Whitney1992}
Hassler Whitney.
\newblock {\em Congruent Graphs and the Connectivity of Graphs}, pages 61--79.
\newblock Birkh{\"a}user Boston, Boston, MA, 1992.

\end{thebibliography}
